\documentclass[12pt, draftclsnofoot, onecolumn]{IEEEtran}
\usepackage[dvips]{color}
\usepackage{tabu}
\usepackage{comment}
\usepackage{todonotes}
\usepackage{epsf}
\usepackage{epsfig}
\usepackage{times}
\usepackage{epsfig}
\usepackage{graphicx}
\usepackage{bbold}
\usepackage{mathtools}
\usepackage{mathrsfs}
\usepackage{amssymb,amsmath}
\usepackage{pdfpages}
\usepackage{epstopdf}
\usepackage{dsfont}
\usepackage{lettrine} 
\usepackage{amsmath,epsfig,amssymb,algorithm,algpseudocode,amsthm,cite,url}

\usepackage[justification=centering]{caption}
\usepackage{subcaption}
\usepackage{bbm}
\usepackage{algorithm}
\usepackage{algpseudocode}
\usepackage{algpseudocode}
\algnewcommand{\LineComment}[1]{\State \(\triangleright\) #1}
\usepackage{csquotes}
\usepackage{amsmath,amssymb}

\DeclareFontFamily{OMX}{yhex}{}
\DeclareFontShape{OMX}{yhex}{m}{n}{<->yhcmex10}{}
\DeclareSymbolFont{yhlargesymbols}{OMX}{yhex}{m}{n}
\DeclareMathAccent{\wideparen}{\mathord}{yhlargesymbols}{"F3}
\usepackage{verbatim}
\usepackage{arcsaops}
\usepackage[english]{babel}
\usepackage{amsmath,amssymb}

\usepackage{verbatim}

\newtheorem{theorem}{\bf Theorem}
\let\oldtheorem\theorem
\renewcommand{\theorem}{\oldtheorem\normalfont}

\newtheorem{proposition}{\bf Proposition}
\let\oldproposition\proposition
\renewcommand{\proposition}{\oldproposition\normalfont}
\newtheorem{lemma}{\bf Lemma}
\let\oldlemma\lemma
\renewcommand{\lemma}{\oldlemma\normalfont}
\newtheorem{example}{Example}
\let\oldexample\example
\renewcommand{\example}{\oldexample\normalfont}

\newtheorem{definition}{\bf Definition}
\let\olddefinition\definition
\renewcommand{\definition}{\olddefinition\normalfont}
\newtheorem{remark}{Remark}
\let\oldremark\remark
\renewcommand{\remark}{\oldremark\normalfont}
\newcommand\semihuge{\fontsize{21}{27}\selectfont}

\DeclareMathOperator*{\argmin}{argmin}
\usepackage{setspace}	
\begin{document}
\title{\semihuge  Caching Meets Millimeter Wave Communications for Enhanced Mobility Management in 5G Networks}\vspace{0em}
\author{
	\authorblockN{Omid Semiari$^{\dag}$, Walid Saad$^{\dag}$, Mehdi Bennis$^\ddag$, and Behrouz Maham$^*$}\\\vspace*{0em}
	\authorblockA{\small $^{\dag}$Wireless@VT, Bradley Department of Electrical and Computer Engineering, Virginia Tech, Blacksburg, USA, Emails: \protect\url{{osemiari,walids}@vt.edu}\\
		\small $^\ddag$ Centre for Wireless Communications, University of Oulu, Finland, Email: \url{bennis@ee.oulu.fi}\\
		$^*$Department of Electrical and Electronic Engineering,
		Nazarbayev University, Astana, Kazakhstan, Email: \url{behrouz.maham@nu.edu.kz}
	}\vspace*{-2em}
}
%
\maketitle
\begin{abstract}

	One of the most promising approaches to overcome the uncertainty and dynamic channel variations of millimeter wave (mmW) communications is to deploy dual-mode base stations that integrate both mmW and microwave ($\mu$W) frequencies. If properly designed, such dual-mode base stations can enhance mobility and handover in highly mobile wireless environments. In this paper, a novel approach for analyzing and managing mobility in joint $\mu$W-mmW networks is proposed. The proposed approach leverages device-level caching along with the capabilities of dual-mode small base stations (SBSs) to minimize handover failures, reduce inter-frequency measurement energy consumption, and provide seamless mobility in emerging dense heterogeneous networks. First, fundamental results on the caching capabilities, including caching probability and cache duration are derived for the proposed dual-mode network scenario. Second, the average achievable rate of caching is derived for mobile users. Moreover, the impact of caching on the number of handovers (HOs), energy consumption, and the average handover failure (HOF) is analyzed. Then, the proposed cache-enabled mobility management problem is formulated as a \emph{dynamic matching game} between mobile user equipments (MUEs) and SBSs. The goal of this game is to find a distributed handover mechanism 
	that, under network constraints on HOFs and limited cache sizes, allows each MUE to choose between: a) executing an HO to a target SBS, b) being connected to the macrocell base station (MBS), or c) perform a transparent HO by using the cached content.
	The formulated matching game inherently captures the dynamics of the mobility management problem caused by HOFs. To solve this dynamic matching problem, a novel algorithm is proposed and its convergence to a two-sided dynamically stable HO policy for MUEs and target SBSs is proved. Numerical results corroborate the analytical derivations and show that the proposed solution will provides significant reductions in both the HOF and energy consumption of MUEs, resulting in an enhanced mobility management for heterogeneous wireless networks with mmW capabilities.
\end{abstract}
\section{Introduction} \label{intro}\vspace{-0cm}
The proliferation of bandwidth-intensive wireless applications such as social networking, high definition video streaming, and mobile TV have drastically strained the capacity of wireless cellular networks. To cope with this traffic increase, several new technologies are anticipated for 5G cellular systems: 1) dense deployment of small cell base stations (SBSs), 2) exploitation of the large amount of available bandwidth at \emph{millimeter wave (mmW)} frequencies, and 3) enabling of \emph{content caching} directly at the user equipments (UEs) to reduce delay and improve quality-of-service (QoS). The dense deployment of SBSs with reduced cell-sizes will boost the capacity of wireless networks by decreasing UE-SBS distance, removing coverage holes, and improving spectral efficiency. Meanwhile, mmW communications will provide high data rates by leveraging directional antennas and transmitting over a large bandwidth that can reach up to $5$ GHz. In addition, exploiting the high storage capacity of the modern smart handhelds to cache the data at the UE increases the flexibility and robustness of resource management, in particular, for \emph{mobile UEs (MUEs)}. In fact, caching allows the network to store the data content in advance, while enabling MUEs to use the cached content when sufficient wireless resources are not available.  

However, dense heterogeneous networks (HetNets), composed of macrocell base stations (MBSs) and SBSs with various cell sizes, will introduce three practical challenges for mobility management. First, MUEs will experience frequent handovers (HOs), while passing SBSs with relatively small cell sizes, which naturally increases the overhead and delay in HetNets. Such frequent HOs will also increase handover failure (HOF), particularly for MUEs that are moving at high speeds \cite{6384454}. In fact, due to the small and disparate cell sizes in HetNets, MUEs will not be able to successfully finish the HO process by the time they trigger HO and pass a target SBS. Second, the inter-frequency measurements that are needed to  discover target SBSs can be excessively power consuming and detrimental for the battery life of MUEs, especially in dense HetNets with frequent HOs. Third, microwave ($\mu$W) frequencies are stringently congested, and thus, frequent HOs may introduce unacceptable overhead and limit the available frequency resources for the static users. In this regard, offloading MUEs from heavily utilized $\mu$W frequencies to mmW frequencies can substantially improve the spectral efficiency at the $\mu$W network.

To address these challenges and enhance mobility management in  HetNets, an extensive body of work has appeared in the literature \cite{1325888,7010527, 6515049,6563279,6587998,7247509,Khan1,7562411,7565107,7354528,7118239,icc17,6736753,6603647,6620380,Gomes2016}. In \cite{1325888}, the authors provide a comprehensive overview on mobility management in IP networks. The authors in \cite{7010527} present different distributed mobility management protocols at the transport layer for future dense HetNets. In \cite{6515049}, an energy-efficient SBS discovery method is proposed for HetNets. The work in \cite{6563279} investigates HO decision algorithms that focus on improving HO between femtocells and LTE-Advanced systems. The work presented in \cite{6587998} overviews existing approaches for vertical handover decisions in HetNets. In \cite{7247509}, the authors study the impact of channel fading on mobility management in HetNets. In addition, the work in \cite{7247509} shows that increasing the sampling period for HO decision decreases the fading impact, while increasing the ping-pong effect. In \cite{Khan1}, the authors propose an HO scheme that takes into account the speed of MUEs to decrease frequent HOs in HetNets. The authors in \cite{7562411} propose an HO scheme that supports soft HO by allowing MUEs to connect with both a macrocell base station (MBS) and SBSs. Furthermore, a distributed mobility management framework is proposed in \cite{7565107} which uses multiple frequency bands to decouple the data and control planes. 

Although interesting, the body of work in \cite{1325888,7010527,6515049,6563279,6587998,7247509,Khan1,7562411,7565107} does not consider mmW communications and caching capabilities for mobility management and solely focuses on HetNets operating over $\mu$W frequencies. In addition, it does not study the opportunities that caching techniques can provide for mobility management. In \cite{7354528}, an HO scheme for mmW networks is proposed in which the MBS acts as an anchor for mmW SBSs to manage control signals. However, \cite{7354528} assumes that line-of-sight (LoS) mmW links are always available and provides no analytical results to capture the directional nature of mmW communications. In \cite{7118239}, the authors propose a resource allocation scheme for hybrid mmW-$\mu$W networks that enhances video streaming by buffering content over mmW links. However, \cite{7118239} does not address any mobility management challenge, such as frequent HOs or HOF. Our early work in \cite{icc17} provided some of the basic insights on mobility management in $\mu$W-mmW networks. However, in contrast to this work, \cite{icc17} solely focuses on an average performance analysis, does not consider dynamic HO problem for multi-MUE scenarios, and does not propose any energy management mechanisms for handling inter-frequency measurements.

Proactive caching for enhancing mobility management has been motivated by the works in \cite{6736753,6603647,6620380,Gomes2016}. In \cite{6736753}, the authors discuss the potential of content caching at either evolved packet core network or radio access network to minimize the traffic overhead at the core network. Moreover, the authors in \cite{6603647} propose a proactive caching framework in which an ongoing IP service can be cached in advance and continuously transferred among different data centers as MUEs move across different cells. In \cite{6620380}, the authors propose a caching framework that stores different parts of a content at different base stations, allowing MUEs to randomly move across different cells and download different cached parts of the original content whenever possible. In addition, in \cite{Gomes2016}, a proactive caching solution is proposed for mobility management by exploiting MUEs' trajectory information. Although interesting, the body of work in \cite{6736753,6603647,6620380,Gomes2016} focuses on adopting protocols that are designed for higher network layers. Moreover, these solutions do not consider caching directly at the MUEs and focus on mobility management at the core network. However, we will show how leveraging high capacity mmW communication complements the notion of caching at MUEs. In addition, caching at MUEs will provide opportunities to perform \emph{transparent} HOs in HetNets, without requiring any data session with a target SBS.

The main contribution of this paper is a novel mobility management framework that addresses critical handover issues, including frequent HOs, HOF, and excessive energy consumption for seamless HO in emerging dense wireless cellular networks with mmW capabilities. In fact, we propose a model that allows MUEs to cache their requested content by exploiting high capacity mmW connectivity whenever available. As such, the MUEs will use the cached content and avoid performing any HO, while passing SBSs with relatively small cell sizes. First, we propose a geometric model to derive tractable, closed-form expressions for key performance metrics, including the probability of caching, cumulative distribution function of caching duration, and the average data rate for caching at an MUE over a mmW link. Moreover, we provide insight on the achievable gains for reducing the number of HOs and the average HOF, by leveraging caching in mmW-$\mu$W networks. Then, we formulate the proposed cache-enabled mobility management framework as a dynamic matching game, so as to provide a distributed solution for mobility management in HetNets, while taking the dynamics of the system into account. To solve the formulated dynamic matching problem, we first show that conventional algorithms such as the deferred acceptance algorithm adopted in \cite{Gale} and \cite{eduard11}, fail to guarantee a dynamically stable HO between MUEs and SBSs. Therefore, we propose a novel distributed algorithm that is guaranteed to converge to a dynamically stable HO policy in dense HetNets. Subsequently, the complexity of the proposed algorithm in terms of signaling overhead is analyzed. Under practical settings, we show that the proposed cache-enabled HO framework can decrease the average HOF rate by up to $45 \%$, even for MUEs with high speeds. In addition, simulation results provide insights on the achievable gains by the proposed distributed algorithm, in terms of reducing energy consumption for cell search, as well as increasing traffic offloads from the $\mu$W frequencies.  


The rest of this paper is organized as follows. Section II presents the system model. Section III presents the analysis for caching in mobility management. Performance analysis of the cache-enabled mobility management is provided in Section IV. Section V formulates the mobility management as a dynamic matching and presents the proposed algorithm. Simulation results are presented in Section VI and conclusions are drawn in Section VII.

\begin{table}[!t]
	\footnotesize
	\centering
	\caption{Variables and notations}\vspace*{-0em}
	\begin{tabular}{|c|c||c|c|}
		\hline
		\bf{Notation} & \bf{Description} & \bf{Notation} & \bf{Description} \\
		\hline
		$K$ & Number of SBSs & $\mathcal{K}$ & Set of SBSs\\
		\hline
		$U$ & Number of MUEs & $\mathcal{U}$ & Set of MUEs\\
		\hline
		$\theta_u$ & Moving angle of MUEs & $v_u$ & Speed of MUEs\\
		\hline
		$p_k$ & Transmit power of SBS $k$ & $B$ & Segment size of video (bits)\\
		\hline
		$\Omega_u$ & Cache size of MUE $u$ & $\Omega_u^{\text{max}}$ & Maximum cache size\\
		\hline
		$t_u^c$ & Caching duration of MUE $u$ & $Q$ & Video play rate\\
		\hline
		$\bar{R}^c(u,k)$ & Average achievable caching rate & $d^c$ & Traversed distance using cached content\\
		\hline
		$\Delta T$ & Time-to-trigger (TTT) & $r^c$ & Traversed distance in caching duration\\
			\hline
		$T_s$ & Inter-frequency cell scanning interval & $t_{\text{MTS}}$ & Minimum time-of-stay (ToS)\\
		\hline
		$\theta_k$ & Beamwidth for SBS $k$ & $E^s$ & Consumed energy per cell search \\
			\hline
		$t_{u,k}$ & Time-of-stay for MUE $u$ at SBS $k$ & $t_{\text{MTS}}$ & Minimum required time-of-stay \\
		\hline
	\end{tabular}\label{tab1}
\end{table}
\section{System Model}
Consider a HetNet composed of an MBS and $K$ SBSs within a set $\mathcal{K}$ distributed uniformly across an area. Each SBS $k \in \mathcal{K}$ can be viewed as a picocell or a femtocell, depending on its transmit power $p_k$. Picocells are typically deployed in outdoor venues while femtocells are relatively low-power and suitable for indoor deployments. The SBSs operate at $\mu$W frequencies that are different than those used by the MBS and, thus, there is no interference between SBSs and the MBS  \cite{power,6515049}. The SBSs are also equipped with mmW front-ends to serve MUEs over either mmW or $\mu$W frequency bands \cite{7503786}. The dual-mode capability allows to integrate mmW and $\mu$W radio access technologies (RATs) at the medium access control (MAC) layer of the air interface and reduce the delay and overhead for fast vertical handovers between both RATs \cite{7503786}. Within this network, we consider a set $\mathcal{U}$ of $U$ MUEs that are distributed randomly and that move across the considered geographical area during a time frame $T$. Each user $u \in \mathcal{U}$ moves in a  random direction $\theta_u \in \left[0,2\pi \right]$, with respect to the $\theta=0$ horizontal angle, which is assumed fixed for each MUE over a considered time frame $T$. In addition, we consider that an MUE $u$ moves with an  average speed $v_u \in \left[v_{\text{min}},v_{\text{max}}\right]$. The MUEs can receive their requested traffic over either the mmW or the $\mu$W band.



\vspace{-.1cm}
\subsection{Channel model}
The large-scale channel effect over mmW frequencies for a link between an SBS $k$ and an MUE $u \in \mathcal{U}$, in dB, is given by\footnote{The free space path loss model in \eqref{pathloss_mmw} has been adopted in many existing works, such as in \cite{Ghosh14}, that carry out real-world measurements to characterize mmW large scale channel effects.}:
\begin{align}\label{pathloss_mmw}
L(u,k)=20\log_{10}\left(\frac{4\pi r_0}{\lambda}\right)
\!+\! 10\alpha \log_{10}\left(\frac{r_{u,k}}{r_0}\right)\!+\!\chi, 
\end{align} 
where \eqref{pathloss_mmw} holds for $r_{u,k}\geq r_{\text{ref}}$, with $r_{\text{ref}}$ and $r_{u,k}$  denoting, respectively, the reference distance and distance between the MUE $u$ and SBS $k$. In addition, $\alpha$ is the path loss exponent, $\lambda$ is the wavelength at carrier frequency $f_c = 73$ GHz over the E-band, due to the low oxygen absorption, and $\chi$ is a Gaussian random variable with zero mean and variance $\xi^2$. The path loss parameters $\alpha$ and $\xi$ will have different values, depending on whether the mmW link is line-of-sight (LoS) or non-LoS (NLoS). Over the $\mu$W frequency band, the path loss model follows \eqref{pathloss_mmw}, however, with parameters that are specific to sub-6 GHz frequencies. 
\begin{figure}[!t]
	\centering
	\centerline{\includegraphics[width=10cm]{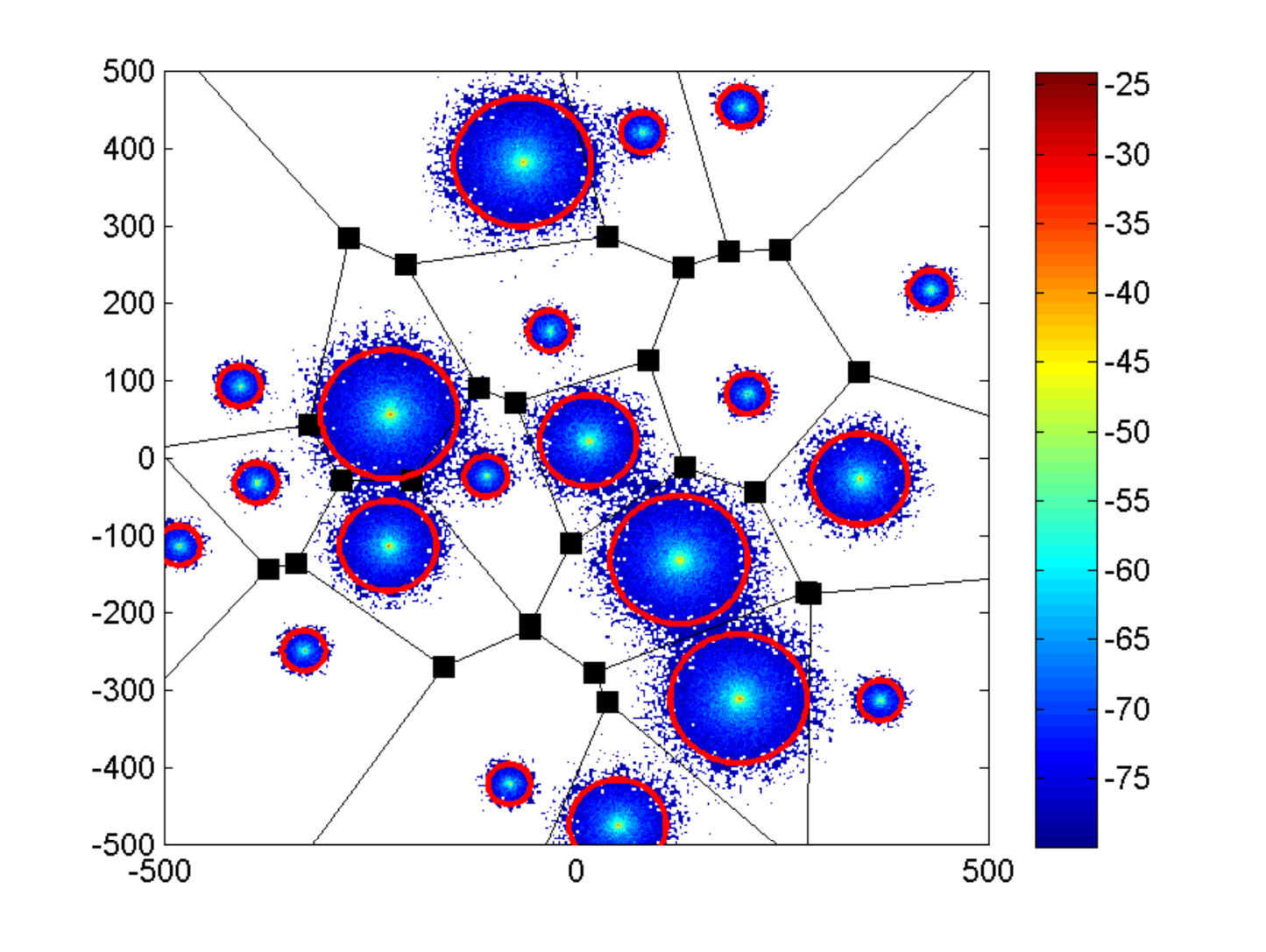}}\vspace{-.5cm}
	\caption{\small SBSs coverage with RSS threshold of $-80$ dB. Red circles show the simplified cell boundaries.}\vspace{-0.2cm}
	\label{Model1}
\end{figure}\vspace{0em}

An illustration of the considered HetNet is shown in Fig. \ref{Model1}. The coverage for each SBS at the $\mu$W frequency is shown based on the maximum received signal strength (max-RSS) criteria with a threshold of $-80$ dB. White spaces in Fig. \ref{Model1} delineate the areas that are covered solely by the MBS. Here, we observe that shadowing effect can adversely increase the ping-pong effect for MUEs. To cope with this issue, the 3GPP standard suggests L1/L3 filtering which basically applies averaging to RSS samples, as explained in \cite{7247509}.  
\subsection{Antenna model and configuration}
To overcome the excessive path loss at the mmW frequency band, the MUEs will be equipped with electronically steerable antennas which allow them to achieve beamforming gains at a desired direction. The antenna gain pattern for MUEs follows the simple and widely-adopted sectorized pattern which is given by \cite{7110547}:
\begin{align}\label{gainMUE}
G(\theta)=\begin{cases}
G_{\text{max}}, &\text{if} \,\,\,\,\,\theta <|\theta_m|,\\
G_{\text{min}}, &\text{otherwise},
\end{cases}
\end{align}
where $\theta$ and $\theta_m$ denote, respectively, the azimuth angle and the antennas' main lobe beamwidth. $G_{\text{max}}$ and $G_{\text{min}}$ denote, respectively, the antenna gain of the main lobe and side lobes. For SBSs, we use a model similar to the  sectorized pattern in \eqref{gainMUE},
\begin{figure}[!t]
	\centering
	\centerline{\includegraphics[width=6cm]{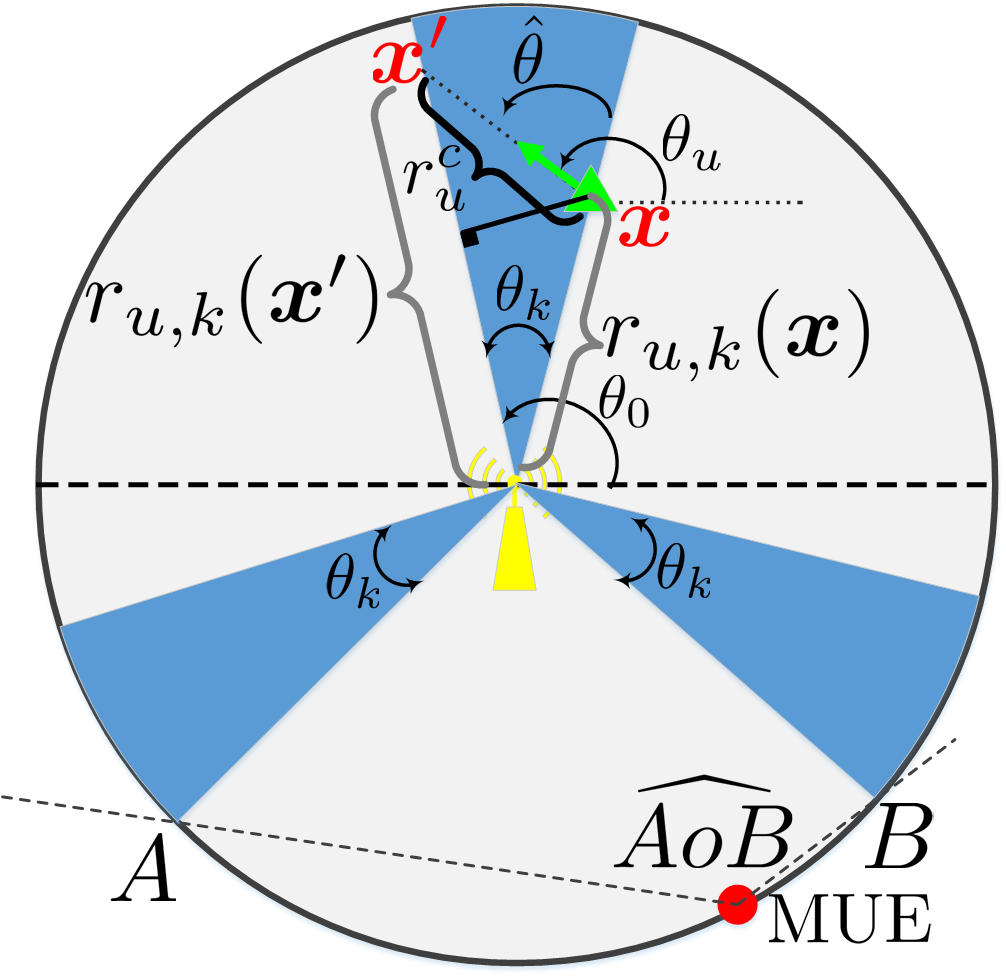}}\vspace{-.2cm}
	\caption{\small Antenna beam configuration of a dual-mode SBS with $N_k=3$. Shaded areas show the mmW beams.}\vspace{0em}
	\label{model2}
\end{figure}
however, we allow each SBS $k$ to form $N_k$ beams, either by using $N_k$ antenna arrays or forming multi-beam beamforming. The beam patten configuration of an SBS $k$ is shown in Fig. \ref{model2}, where $N_k=3$ equidistant beams in $\theta \in \left[0, 2\pi \right]$ are formed. To avoid the complexity and overhead of beam-tracking for mobile users, the direction of the SBSs' beams in azimuth is fixed. In fact, an MUE can connect to an SBS $k$ over a mmW link, if the MUE traverses the area covered by the $k$'s mmW beams.  It is assumed that for a desired link between an SBS $k$ and an MUE $u$, the overall transmit-receive gain is $\psi_{u,k}=G_{\text{max}}^2$. 
\vspace{-1em}
\subsection{Traffic model}\vspace{-.1cm}
Video streaming is one of the wireless services with most stringent quality-of-service (QoS) requirement. Meeting the QoS demands of such services is prone to the delay caused by frequent handovers in HetNets. In addition, HOFs can significantly degrade the performance by making frequent service interruptions. Therefore, our goal is to enhance mobility management for MUEs that request video or streaming traffic. Each video content is partitioned into small segments, each of size $B$ bits. The network incorporates caching to transmit incoming video segments to an MUE, whenever a high capacity mmW connection is available. In fact, high capacity mmW connection, if available, allows to cache a large portion or even the entire video in a very short period of time. We define the cache size of $\Omega_u(k)$ for an arbitrary MUE $u$, associated with an SBS $k$, as the number of video segments that can be cached at MUE $u$ as follows:
\begin{align}\label{traffic1}
\Omega_u(k) = \min\Big\lbrace\bigg\lfloor \frac{\bar{R}^c(u,k)t_u^c}{B}\bigg\rfloor, \Omega_u^{\text{max}}\Big\rbrace,
\end{align}
where $\lfloor . \rfloor$ and $\min\lbrace . ,. \rbrace$ denote, respectively, the floor and minimum operands and $\Omega_u^{\text{max}}$ is the maximum cache size. In addition, $t_u^c$ is the \emph{caching duration} which is equal to the time needed for an MUE $u$ to traverse the mmW beam of its serving SBS. Considering the small green triangle in Fig. \ref{model2} as the current location of an MUE crossing a mmW beam, the caching duration will be $t_u^c=r_u^c/v_u$ where $r_{u}^c$ is the distance traversed across the mmW beam. Moreover, $\bar{R}^c(u,k)$ is the \emph{average achievable rate} for the MUE $u$ during $t_u^c$. Given $\Omega_u(k)$ and the  video play rate of $Q$, specified for each video content, the distance an MUE $u$ can traverse with speed $v_u$, while playing the cached video content will be 
\begin{align}\label{traffic2}
d^c(u,k) = \frac{\Omega_u(k)}{Q}v_u.
\end{align}
In fact, the MUE can traverse a distance $d^c(u,k)$ by using the cached video content after leaving its serving cell $k$, without requiring an HO to any of the target cells. Meanwhile, the location information and control signals, such as paging, can be handled by the MBS during this time. As we discuss in details, such caching mechanism will help MUEs to avoid redundant cell search and HOs, resulting in an efficient mobility management in dense HetNets.
\vspace{-1em}
\subsection{Handover procedure and performance metrics}
The HO process in the 3GPP standard proceeds as follows: 1) Each MUE will do a cell search every $T_s$ seconds, which can be configured by the network or directly by the MUEs, 2) If any target cell offers an RSS plus a hysteresis that is higher than the serving cell, even after L1/L3 filtering of input RSS samples, the MUE will wait for a time-to-trigger (TTT) of $\Delta T$ seconds to measure the average RSS from the target cell, 
3) If the average RSS is higher than that of the serving SBS during TTT, the MUE triggers HO and sends the measurement report to its serving cell. 
The averaging over the TTT duration will reduce the ping-pong effect resulting from instantaneous CSI variations, and 4) HO will be executed after the serving SBS sends the HO information to the target SBS.

In our model, we modify the above HO procedure to leverage the caching capabilities of MUEs during mobility. Here, we let each MUE $u$ dynamically determine $T_s$, depending on the cache size $\Omega_u$, the video play rate $Q$, and the MUE's speed $v_u$. That is, an MUE $u$ is capable of muting the cell search while $\Omega_u/Q$ is greater than $\Delta T$, which enables it to have $\Delta T$ seconds to search for a target SBS before the cached content runs out. 

Next, we consider the HOF as one of the key performance metrics for any HO procedure. One of the main reasons for the potential increase in HOF in HetNets is due to the relatively small cell sizes, compared to MBS coverage. In fact, HOF is typical if the time-of-stay (ToS) for an MUE is less than the minimum ToS (MTS) required for performing a successful HO. That is,
\begin{align}\label{HOF1}
\gamma_{\text{HOF}}(u,k) = \begin{cases}
1, \,\,\,\,\,\, \text{if}\,\,\, t_{u,k}<t_{\text{MTS}},\\
0, \,\,\,\,\,\, \text{otherwise},
\end{cases}
\end{align}
where $t_{u,k}$ is the ToS for MUE $u$ to pass across SBS $k$ coverage. Although a short ToS may not be the only cause for HOFs, it becomes very critical within an ultra dense small cell network that encompasses MUEs moving at high speeds \cite{3gpp}.

To search the $\mu$W carrier for synchronization signals and decode the broadcast channel (system information)
of the detected SBSs, the MUEs have to spend an  energy $E^s$ per each cell search \cite{6515049}. Hence, the total energy consumed by an MUE for cell search during time $T$  will be
\begin{align}\label{P1}
E_{\text{total}}^s = E^s  \frac{T}{T_s}.
\end{align}
Note that increasing $T_s$ reduces $E_{\text{total}}^s$ which is desirable. However, less frequent scans will be equivalent to less HOs to SBSs. Therefore, there is a tradeoff between reducing the consumed power for cell search and maximizing traffic offloads from the MBS to SBSs. Content caching will allow increasing $T_s$, while maintaining traffic offloads from the MBS.

Next, we propose a geometric framework to analyze the caching opportunities, in terms of the caching duration $t^c$, and the average achievable rate $\bar{R^c}$, for MUEs moving at random directions in joint mmW-$\mu$W HetNets. \vspace{-.1cm}
\section{Analysis of Mobility Management with Caching Capabilities}
In this section, we first investigate the probability of serving an arbitrary MUE over mmW frequencies by a dual-mode SBS. \vspace{-1em}
\subsection{Probability of mmW coverage}
In Fig. \ref{model2}, the small circle represents the intersection of an MUE $u$'s trajectory with the coverage area of an SBS $k$. In this regard, $\mathbbm{P}^c_{k}(N_k,\theta_k)$ represents the probability that MUE $u$ with a random direction $\theta_u$ and speed $v_u$ crosses the mmW coverage areas of SBS $k$. From Fig. \ref{model2}, we observe that the MUE will pass through the area within mmW coverage only if the MUE's direction is inside the angle $\widehat{AoB}$. Hence, we can state the following.

\begin{theorem}\label{prop1}
	If an SBS $k$ has formed a mmW beam pattern with $N_k\geq 2$ main lobes, each with a beamwidth $\theta_k>0$, the probability of content caching will be given by:
	\begin{align}\label{prop1eq}
	\mathbbm{P}^c_{k}(N_k,\theta_k) =\left[\frac{N_k \theta_k}{2\pi}\right] \!+\! \left[1-\frac{N_k \theta_k}{2\pi}\right]\left[\frac{1}{2}\left(1-\frac{1}{N_k}\right)+\frac{\theta_k}{4\pi}\right].	
	\end{align}
\end{theorem}
\begin{proof}
Due to the equidistant beams, we have 
\begin{align}\label{app1-1}
\widehat{AoB} = \frac{1}{2}\wideparen{AB} =\frac{1}{2}\left[2\pi-\wideparen{AoB}\right]= \frac{1}{2}\left[2\pi - \left(\frac{2\pi}{N_k}-\theta_k\right)\right]
=\left(1-\frac{1}{N_k}\right)\pi + \frac{\theta_k}{2}.
\end{align}
Given that an arbitrary MUE can enter the circle in Fig. \ref{model2} at any direction, this MUE will be instantly covered by mmW with probability $\mathbbm{P}(\boldsymbol{x}_u \in \mathcal{A}) = \frac{N_k\theta_k}{2\pi}$, where $\mathcal{A} \subset \mathbbm{R}^2$ denotes the part of circle's perimeter that overlaps with mmW beams. Therefore, 
\begin{align}\label{app1-2}
\mathbb{P}^c_{k}(N_k,\theta_k) = \mathbb{P}(\boldsymbol{x}_u \in \mathcal{A}) + \left[1-\mathbb{P}(\boldsymbol{x}_u \in \mathcal{A})\right]\frac{1}{2\pi}\widehat{AoB},
\end{align}
where \eqref{app1-2} results from the fact that $\theta_u \sim U\left[0, 2\pi\right]$. Therefore, from \eqref{app1-1} and \eqref{app1-2}, the probability of crossing a mmW beam follows \eqref{prop1eq}.
\end{proof}
We can verify \eqref{prop1eq} by considering an example scenario with $N_k=3$ and $\theta_k = \frac{2\pi}{3}$. For this example, \eqref{prop1eq} results in
$\mathbbm{P}^c_{k}(N_k,\theta_k) =1$ which correctly captures the fact that the entire cell is covered by mmW beams. 
\vspace{-0cm}
\subsection{Cumulative distribution function of the caching duration}
To enable an MUE to use the cached content while not being associated to an SBS, it is critical to analyze the distribution of caching duration $t^c$ for an arbitrary MUE with a random direction and speed. In this regard, consider the small green triangle in Fig. \ref{model2}, which represents the location of an arbitrary MUE $u$, $\boldsymbol{x_u} = \left(x_u, y_u\right) \in \mathbbm{R}^2$, crossing a mmW beam. First, we note that the geometry of the mmW beam of any given SBS can be defined by the location of the SBS, as well as the sides of the beam angle. Without loss of generality, we assume that the SBS of interest is located at the center, such that $\boldsymbol{x}_k = (0,0)$. Therefore, the two sides of the beam angle will be given by \vspace{-1em}
\begin{align}\label{beamangle1}
y = x\tan(\theta_0-\theta_k), y= x\tan(\theta_0), \,\,\,\, x>0.
\end{align}
Assuming that
the MUE $u$ is currently located on the angle side $x = y\cos(\theta_0-\theta_k)$, as shown by the small triangle in Fig. \ref{model2}, then $\theta_0$ in \eqref{beamangle1} will be $\theta_0 = \arccos\left(\frac{x_u}{r_{u,k}(\boldsymbol{x_u})}\right)+\theta_k$, where $r_{u,k}(\boldsymbol{x})=\sqrt{x_u^2 + y_u^2}$. Hereinafter, we will use the parameter $\theta_0$ to simplify our  analysis. Let $F_{t^c}(.)$ be the cumulative distribution function (CDF) of the caching duration $t^c$. Thus,
\begin{align}\label{caching1}
F_{t_u^c}(t_0) = \mathbbm{P}(t_u^c \leq t_0) = \mathbbm{P}(r_u^c \leq v_u t_0),
\end{align}
where $r_u^c$ is the distance that MUE $u$ will traverse across the mmW beam, as shown in Fig. \ref{model2}. Given the location of MUE $\boldsymbol{x}_u$, the minimum possible distance to traverse, $r_u^{\text{min}}$, is
\begin{align}\label{caching2}
r_u^{\text{min}} = \frac{\big|x_u\tan\theta_0 -y_u\big|}{\sqrt{1+ \tan^2\theta_0}}.
\end{align}
In fact, \eqref{caching2} gives the distance of the point $\boldsymbol{x}_u$ from the beam angle side $y = x\tan(\theta_0)$. If $r_u^{\text{min}} >v_u t_0$, then $F_{t_u^c}(t_0) = 0$. Therefore, for the remainder of this analysis we consider $r_u^{\text{min}} \leq v_u t_0$. Next, let $\boldsymbol{x}_u'$ denote the intersection of the MUE's path with line $y = x\tan(\theta_0)$. It is easy to see that $\boldsymbol{x}_u' = \left(x_u + r^c_u \cos\theta_u,y_u + r^c_u\sin\theta_u\right)$. Hence, $y_u + r^c_u\sin\theta_u = \left[x_u + r^c_u \cos\theta_u\right]\tan\theta_0$, and $r^c_u$, i.e., the distance that MUE $u$ traverses during the caching duration $t^c$, is given by: 
\begin{align}\label{caching3}
r^c_u =v_ut_u^c= \frac{y_u-x_u\tan\theta_0}{\tan\theta_0 \cos \theta_u-\sin \theta_u}.
\end{align}

Next, from \eqref{caching1} and \eqref{caching3}, the CDF can be written as  
\begin{align}\label{caching4}
F_{t_u^c}(t_0)=\mathbb{P}\left(\frac{y_u-x_u\tan\theta_0}{\tan\theta_0 \cos \theta_u-\sin \theta_u} \leq v_u t_0\right).
\end{align}
Using the geometry shown in Fig. \ref{model2}, we find the CDF of the caching duration as follows:
\begin{lemma}\label{prop2}
	The CDF of the caching duration, $t^c$, for an arbitrary MUE $u$ with speed $v_u$ is given by
	\begin{align}\label{caching5}
	F_{t^c}(t_0)= \frac{1}{\pi-\theta_k}\bigg(\arccos\left(\frac{r_u^{\text{min}}}{v_ut_0}\right)+\min\bigg\lbrace \arccos\left(\frac{r_u^{\text{min}}}{r_{u,k}(\boldsymbol{x})}\right), \arccos\left(\frac{r_u^{\text{min}}}{v_ut_0}\right)\bigg\rbrace\bigg).
	\end{align}
\end{lemma}
\begin{proof}
	From \eqref{caching1}, $F_{t^c}(t_0)=\mathbb{P}(r_u^c\leq v_u t_0)$. To find this probability, we note that $r_u^c\leq v_u t_0$ if MUE moves between two line segments of length $v_ut_0$ that connect MUE to line $y=x\cos\theta_0$. Depending on $r_{u,k}(\boldsymbol{x})$, the intersection of line segment with $y=x\cos\theta_0$ may have one or two solutions. In case of two intersection points, the two line segments will make two equal angles with the perpendicular line from $\boldsymbol{x}_u$, to $y=x\cos\theta_0$, which each is obviously equal to $\pi-(\pi/2-\theta_k)-\hat{\theta}=\pi/2+\theta_k-\hat{\theta}=\arccos\left(\frac{r_u^{\text{min}}}{v_ut_0}\right)$. Therefore, 
	\begin{align}\label{aaaeq1}
	F_{t^c}(t_0)=\frac{2}{\pi-\theta_k}\arccos\left(\frac{r_u^{\text{min}}}{v_ut_0}\right).
	\end{align}
	In fact, $\theta_u$ must be within a range of $\pi-\theta_k$ for $r_u^c\leq v_ut_0$ to be valid. 
	Now, if this angle is greater than $\pi/2-\theta_k$, only one intersection point exists. Equivalently,
	\begin{align}\label{aaaeq2}
	\!\!\!\!F_{t^c}(t_0)\!=\!\frac{1}{\pi-\theta_k}\!\bigg(\!\!\arccos\left(\frac{r_u^{\text{min}}}{v_ut_0}\right)\!+\!\arccos\left(\frac{r_u^{\text{min}}}{r_{u,k}(\boldsymbol{x})}\right)\!\!\bigg).
	\end{align}
	Integrating \eqref{aaaeq1} and \eqref{aaaeq2}, the CDF for caching duration can be written as \eqref{caching5}.
\end{proof}
The CDF of $t^c$ is shown in Fig. \ref{CDF1} for different MUE distances from the serving SBS. Fig. \ref{CDF1} shows that as the MUE is closer to the SBS, $t^c$ takes smaller values with higher probability which is expected, since the MUE will traverse a shorter distance to cross the mmW beam.\vspace{-1em}
\section{Performance Analysis of the Proposed Cache-enabled Mobility Management Scheme}\label{sec:IV}
\begin{figure}[!t]
	\centering
	\centerline{\includegraphics[width=8cm]{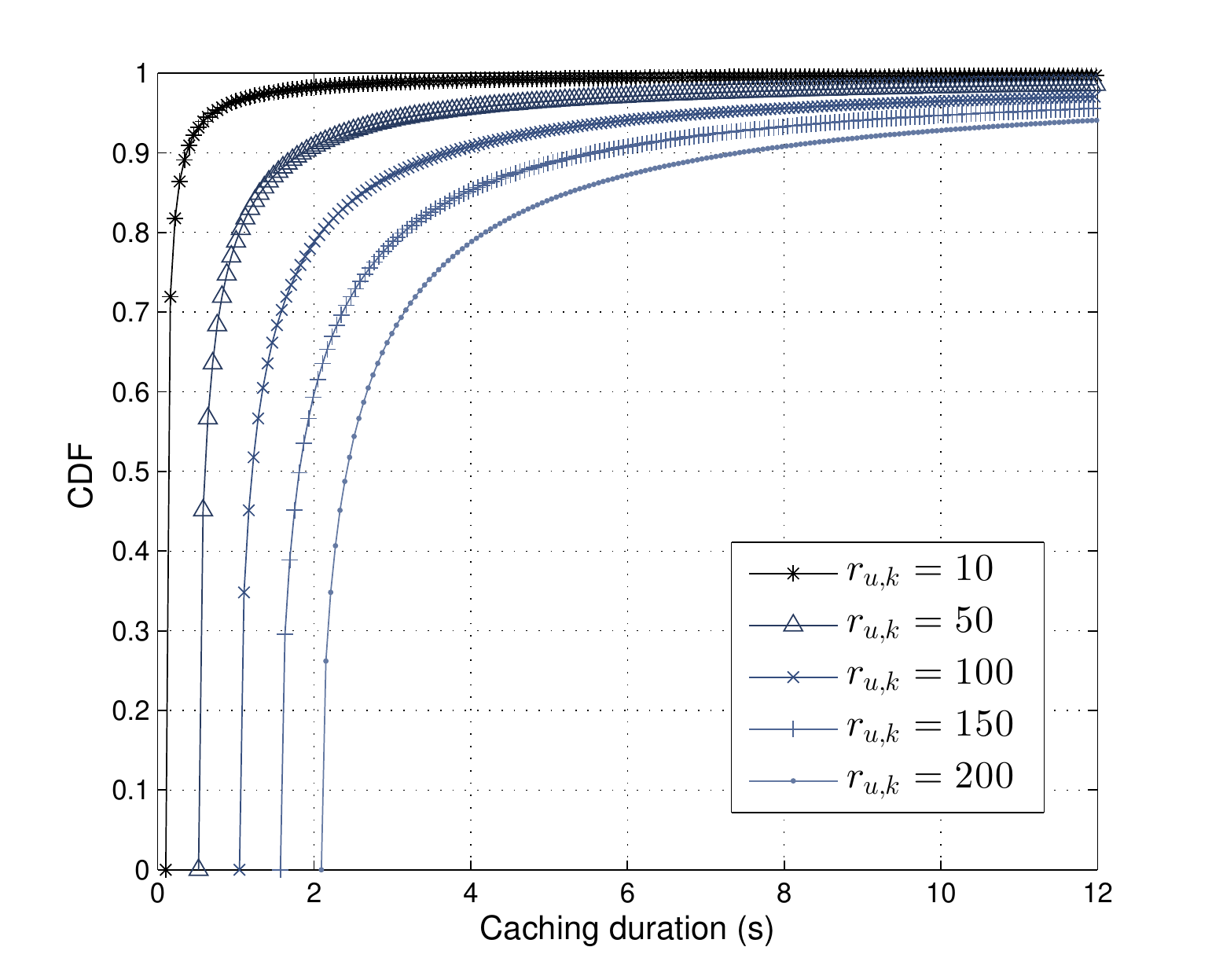}}\vspace{-.2cm}
	\caption{\small CDF of caching duration $t^c$.}\vspace{-0.3cm}
	\label{CDF1}
\end{figure}
Next, we analyze the average achievable rate for content caching, for an MUE with speed $v_u$, direction $\theta_u$, and initial distance $r_{u,k}(\boldsymbol{x})$ from the serving dual-mode SBS. In addition, we evaluate the impact of caching on mobility management. For this analysis, we ignore the shadowing effect and only consider distance path loss. \vspace{-1em}
\subsection{Average Achievable Rate for Caching}
The achievable rate of caching is given by:
\begin{align}\label{rate1}
\!\!R^c(u,k) = \frac{1}{v_ut_u^c}\int_{r_{u,k}(\boldsymbol{x})}^{r_{u,k}(\boldsymbol{x}')}\!\!\!\!\!\!w\log\left(1+ \frac{\beta P_t\psi r_{u,k}^{-\alpha}}{wN_0}\right)dr_{u,k},
\end{align}
where $\beta =(\frac{\lambda}{4\pi r_0})^2r_0^{\alpha}$. The integral in \eqref{rate1} is taken over the line with length $r_u^c$ that connects the MUE location  $\boldsymbol{x}$ to $\boldsymbol{x}'$, as shown in Fig. \ref{model2}. With this in mind, we can find the average achievable rate of caching $\bar{R}^c$ as follows.
\begin{theorem}
	The average achievable rate for an MUE $u$ served by an SBS $k$, $\bar{R}^c(u,k)$, is:
	\begin{align}\label{rate2}
	\!\!\!\!\!\!	&\bar{R}^c(u,k) = \mathbb{P}^c_{k}(N_k,\theta_k)R^c(u,k),\\\label{rate18}
	&=\delta_2 \int_{f(\theta_k)}^{f(0)}\frac{1}{f^2(\theta)}\log\left(1+\delta_1f^{\alpha}(\theta)\right)df(\theta),\\\notag
	& \stackrel{\text{(a)}}{=} \frac{\delta_2}{\ln(2)}\! \bigg[2\sqrt{\delta_1}\arctan(\sqrt{\delta_1}f(\theta_k))\!-\!\frac{\ln(\delta_1f^2(\theta_k)+1)}{f(\theta_k)},\\\label{rate19}
	&-\! 2\sqrt{\delta_1}\arctan(\sqrt{\delta_1}f(0))\!+\!\frac{\ln(\delta_1f^2(0)+1)}{f(0)}\bigg],
	\end{align}
	where $\delta_1 = \frac{\beta P_t\psi}{wN_0}\left[r_{u,k}(\boldsymbol{x})\sin\hat{\theta}\right]^{-\alpha}$. Moreover, $\delta_2=wr_{u,k}(\boldsymbol{x})\sin\hat{\theta}\mathbb{P}^c_{k}(N_k,\theta_k)/v_ut^c$, and $\hat{\theta} = \theta_u-\theta_0+\theta_k$. For (a) to hold, we set  $\alpha=2$ which is a typical value for the path loss exponent of LoS mmW links \cite{Ghosh14}.
\end{theorem} \vspace{-1em}
\begin{proof}
	Theorem \ref{prop1} implies that with probability $1-\mathbb{P}^c_{k}(N_k,\theta_k)$, only $\mu$W coverage is available for an MUE. Therefore, the average achievable rate for caching over the mmW frequencies is given by \eqref{rate2}. To simplify \eqref{rate2}, we have
	\begin{align}\label{rate20}
	r_{u,k}\cos\theta = r_{u,k}(\boldsymbol{x}) + r_u\cos\hat{\theta}, \,\, 	r_{u,k}\sin\theta = r_u\sin\hat{\theta},
	\end{align}
	where $\hat{\theta} = \theta_u-\theta_0+\theta_k$ and $\theta$ is an angle between the line connecting MUE to SBS, ranging from $0$ to $\theta_k$. Moreover, $r_u$ is the current traversed distance, with $r_u = r_u^c$ once the MUE reaches $\boldsymbol{x}'$ by the end of caching duration, as shown in Fig. \ref{model2}. From \eqref{rate20}, we find $r_{u,k} = r_{u,k}(\boldsymbol{x})\sin\hat{\theta}/\sin(\hat{\theta}-\theta)$. By changing the integral variable $r_u$ to $\theta$, we can write \eqref{rate2} as 
	\begin{align}\label{rate21}
	\bar{R}^c(u,k) = \delta_2\int_{0}^{\theta_k}\log\left(1+ \delta_1\sin^{\alpha}(\hat{\theta}-\theta)\right)\frac{\cos(\hat{\theta}-\theta)}{\sin^2(\hat{\theta}-\theta)}d\theta,\vspace{-.1cm}
	\end{align}
	where $\delta_1 = \beta P_t\psi(r_{u,k}(\boldsymbol{x})\sin\hat{\theta})^{-\alpha}/wN_0$ and $\delta_2=wr_{u,k}(\boldsymbol{x})\sin\hat{\theta}\mathbb{P}^c_{k}(N_k,\theta_k)/v_ut^c$. Next, we can directly conclude \eqref{rate18} from \eqref{rate21} by substituting $f(\theta) = \sin(\hat{\theta}-\theta)$ in \eqref{rate21}. For $\alpha=2$, which is a typical value for the path loss exponent for LoS mmW links, \eqref{rate18} can be simplified into \eqref{rate19} by taking the integration by parts in \eqref{rate18}.
\end{proof}\vspace{-1em}
\subsection{Achievable gains of caching for mobility management}
From \eqref{traffic1}, \eqref{traffic2}, and \eqref{rate19}, we can find $d^c(u,k)$ which is the distance that MUE $u$ can traverse, while using the cached video content. On the other hand, by having the average inter-cell distances in a HetNet, we can approximate the number of SBSs that an MUE can pass over distance $d^c(u,k)$. Hence, the average number of SBSs that MUE is able to traverse without performing cell search for HO is
\begin{align}\label{perf1}
\eta\approx\bigg\lfloor \frac{\mathbb{E}\left[d^c(u,k)\right]}{l}\bigg\rfloor, 
\end{align}
where the expected value is used, since $d^c(u,k)$ is a random variable that depends on  $\theta_u$. Moreover, $l$ denotes the average inter-cell distance. Here, we note that 
\begin{align}\label{perf2}
\mathbb{E}\left[d^c(u,k)\right] = \int_{0}^{\infty}\left(1-F_{t_u^c}(v_ut)\right)dt,	
\end{align}
where $F_{t^c}(.)$ is derived in Lemma \ref{prop2}. We note that \eqref{perf2} is the direct result of writing an expected value in terms of CDF. Based on the definition of $\eta$ in \eqref{perf1} and considering that the inter-frequency energy consumption linearly scales with the number of scans, we can make the following observation. \vspace{-.5em} 
\begin{remark}
	The proposed caching scheme will reduce the average energy consumption $E^s$ for inter-frequency cell search by a factor of $1/\eta$ with $\eta$ being defined in \eqref{perf1}.
\end{remark}
Furthermore, from the definition of $\gamma_{\text{HOF}}$ in \eqref{HOF1}, we can define the probability of HOF as $\mathbb{P}(D_{u,k}<v_ut_{\text{MTS}})$ \cite{6849322}, where $D_{u,k}=t_{u,k}/v_u$, and $t_{u,k}$ is the ToS. To compute the HOF probability, we use the probability density function (PDF) of a random chord length within a circle with radius $a$, as follows:
\begin{align}\label{perf3}
f_D(D)=\frac{2}{\pi\sqrt{4a^2-D^2}},
\end{align}
where \eqref{perf3} relies on the assumption that one side of the chord is fixed and the other side is determined by choosing a random $\theta \in \left[0,\pi\right]$. This assumption is in line with our analysis as shown in Fig. \ref{model2}. Using \eqref{perf3}, we can find the probability of HOF as follows:
\begin{align}
\mathbbm{P}(D_{u,k}<v_ut_{\text{MTS}})=\int_{0}^{v_ut_{\text{MTS}}} \frac{2}{\pi\sqrt{4a_{k}^2-D^2}}dD=\frac{2}{\pi}\arcsin\left(\frac{v_ut_{\text{MTS}}}{2a_k}\right).\label{idontknow1}
\end{align}
In fact, $\gamma_{\text{HOF}}$ is a binomial random variable whose probability of success  depends on the MUE's speed, cell radius, and $t_{\text{MTS}}$. Hence, by reducing the number of HOs by a factor of $1/\eta$, the proposed scheme will reduce the expected value of the sum $\sum \gamma_{\text{HOF}}$, taken over all SBSs that an MUE visits during the considered time $T$. 

Thus far, the provided analysis are focused on studying the caching opportunities for the mobility management in single-MUE scenarios. However, in practice, the SBSs can only serve a limited number of MUEs simultaneously. Therefore, an HO decision for an MUE is affected by the decision of the other MUEs. In this regard, we propose a cache-enabled mobility management framework to capture the inter-dependency of HO decisions in dynamic multi-MUE scenarios. 
\vspace{-0em}
\section{Dynamic Matching for Cache-enabled Mobility Management}\label{sec:V}

Within the proposed mobility management scenarios, the MUEs have a flexibility to perform either a vertical or horizontal HO, while moving to their chosen target cell. Additionally, as elaborated in Section \ref{sec:IV}, caching enables MUEs to skip a certain HO, depending on the cache size $\Omega$. In fact, there are three HO actions possible for an arbitrary MUE that is being served by an SBS: 1) Execute an HO for a new assignment with a target SBS, 2) Use the cached content and mute HO, 3) Perform an HO to the MBS. Similarly, an MUE assigned to the MBS can decide whether to handover to an SBS, use cached content, or stay connected to the MBS.

Our next goal is to maximize possible handovers to the SBSs in order to increase the traffic offload from the MBS, subject to constraints on the HOF, SBSs' quota, and limited cache sizes.~With this in mind, our goal is to find an HO policy $\boldsymbol{\zeta}$ for MUEs and target BSs\footnote{For brevity, if not specified, we refer to a base station (BS) as either an SBS $k \in \mathcal{K}$ or the MBS $k_0$.} that  satisfies: 
\begin{subequations}
	\begin{IEEEeqnarray}{l}
		\argmin_{\boldsymbol{\zeta}} \sum_{u \in \mathcal{U}} \zeta(u,k_0),\label{opt:a}\\
		\!\!\!\!\!\!\!\!\!\!\!\!\!\text{s.t.}\,\,\,\,\,
	  	\mathbbm{P}\left(\sum_{k \in \mathcal{K}}\zeta(u,k)D_{u,k} < v_ut_{\text{MTS}} \right)\leq P_u^{\text{th}},  \label{opt:b}\\	
	\left[1-\sum_{k \in \mathcal{K}'}\zeta(u,k)\right]T_s \leq 	\frac{\Omega_u}{Q}, \label{opt:c}\\
		\sum_{k \in \mathcal{K}'}\zeta(u,k) \leq 1,\label{opt:d}\\
		\sum_{u \in \mathcal{U}}\zeta(u,k) \leq U_k^{\text{th}}, \hfill \forall k \in \mathcal{K},\label{opt:e}\\
		\zeta(u,k) \in \{0,1\},\label{opt:f}
	\end{IEEEeqnarray}
\end{subequations}
where $\mathcal{K}'=\mathcal{K}\cup \{k_0\}$ and $\boldsymbol{\zeta}$ is a vector of binary elements $\zeta(u,k) \in \{0,1\}$. In fact, a variable $\zeta(u,k)=1$, if MUE $u$ is chosen to execute an HO to the target cell $k$, otherwise, $\zeta(u,k)=0$. Constraints \eqref{opt:b}-\eqref{opt:f} must hold for all $u \in \mathcal{U}$. In fact, the objective in \eqref{opt:a} is to minimize the number of MUEs associated with MBS $k_0$. \eqref{opt:b} ensures that once an MUE $u$ is assigned to an SBS, i.e. $\sum_{k \in \mathcal{K}}\zeta(u,k)=1$, the probability of HOF must be less than a threshold $P_u^{\text{th}}$, determined based on the QoS requirement of the MUE $u$'s service. Constraint \eqref{opt:c} ensures that if an MUE $u$ is not assigned to any SBS nor the MBS, there will be enough cached video segments for the next $T_s$ time duration. Moreover, constraints \eqref{opt:d} and \eqref{opt:e} indicate, respectively, that each MUE can be assigned to at most one BS and each SBS can serve maximum $U_k^{\text{th}}$ MUEs simultaneously. 

We note that using \eqref{idontknow1}, we can rewrite \eqref{opt:b} as $\sum_{k \in \mathcal{K}}\frac{2}{\pi}\arcsin\left(\frac{v_ut_{\text{MTS}}}{2a_k}\right)\zeta(u,k)\leq P_u^{\text{th}}$, which is a linear constraint. Hence, the posed problem in \eqref{opt:a}-\eqref{opt:f} is an integer linear programming (ILP), and thus, it is NP-hard. Although an approximation algorithm can be employed to solve \eqref{opt:a}-\eqref{opt:f}, centralized algorithms are not scalable and typically introduce latency which is not desired for real-time applications such as streaming for mobile users. Moreover, these solutions will typically rely on the current network instances, such as the location, speed and cache size of the MUEs, and, hence, they  fail to capture the dynamics of the system. To show this, we consider two critical scenarios, shown in Fig. \ref{scenarios}, as follows:
\begin{figure}[!t]
	\centering
	\centerline{\includegraphics[width=12cm]{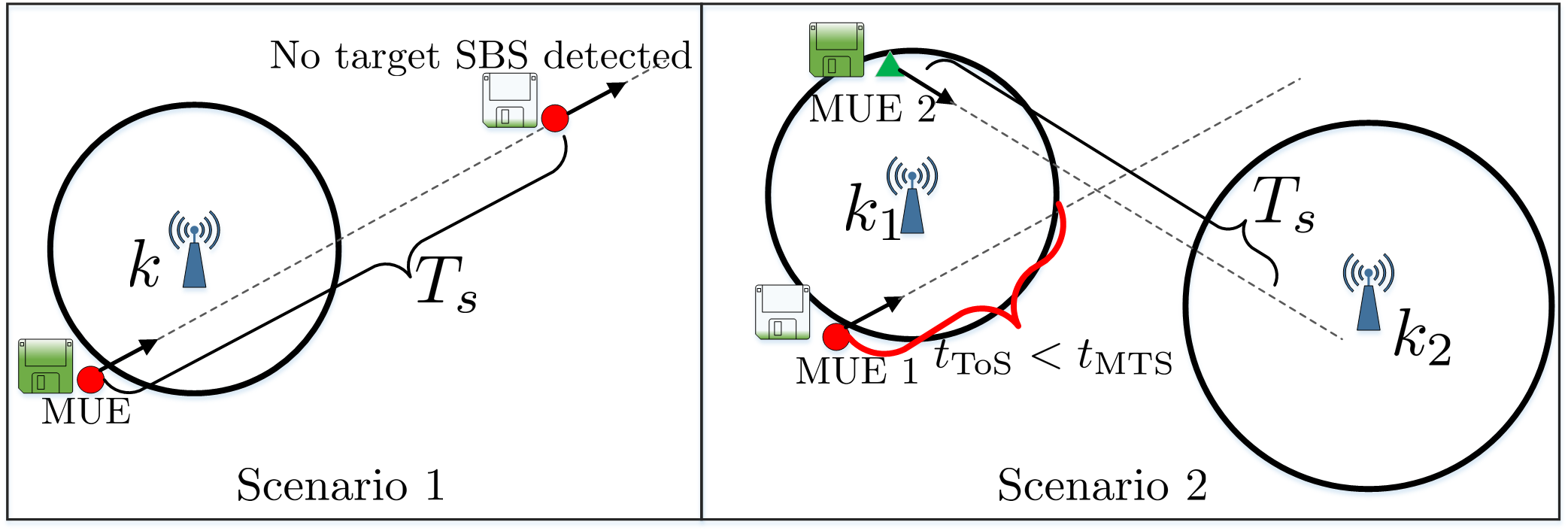}}\vspace{-0.1cm}
	\caption{\small Two dynamic HO scenarios for cache-enabled mobile users.}\vspace{-0cm}
	\label{scenarios}
\end{figure}

\textbf{Illustrative Scenario 1:} Consider a feasible solution for \eqref{opt:a}-\eqref{opt:f}, where an MUE $u$ is not assigned to the target SBS $k$ and will use the cached content for the next $T_s$ time duration, as shown in scenario 1 of Fig. \ref{scenarios}. However, the MUE has to be assigned to the MBS after $T_s$, since eventually no target SBS is detected. Alternatively, the MUE could be assigned to $k$ initially and fill up the cache, while later, it could use the saved cached content to reach the next target cell without requiring to be assigned to the MBS.

\textbf{Illustrative Scenario 2:} Consider a feasible solution for \eqref{opt:a}-\eqref{opt:f} which assigns an arbitrary MUE 1, in Fig. \ref{scenarios}, to a target SBS $k_1$. If there are not enough cached contents for the MUE 1 to move to the next SBS and at the same time HO fails, the MUE has to be assigned to the MBS as shown in scenario 2 of Fig. \ref{scenarios}. Alternatively, we could assign MUE 2 with a large cache size to the SBS $k_1$ such that in case of an HOF, the MUE 2 can reach the next target SBS $k_2$ by using its available cached contents.

These examples show that taking into account the future network information, such the estimated distance from the next target SBS, is imperative to effectively maximize the traffic offloads from the MBS. Therefore, an efficient HO policy must take into account  post-handover scenarios that may occur due to the HOFs. In this regard, we propose a framework based on \emph{dynamic matching theory} \cite{harvard16} which allows effective mobility management, as presented in \eqref{opt:a}-\eqref{opt:f}, while capturing the future network instances, such as the cache size, MUEs' trajectory, and the topology of the network. Next, we present the fundamentals of matching theory and explain how the proposed problem can be formulated as a dynamic matching problem.\vspace{-1em}
\subsection{Handover as a matching game: Preliminaries}
Matching theory is a mathematical framework that provides polynomial time solutions for combinatorial assignment problems such as \eqref{opt:a}-\eqref{opt:f} \cite{Gale}. In a static form, a matching game is defined as a two-sided assignment problem between two disjoint sets of players in which the players of each set are interested to be matched to the players of the other set, according to their preference profiles. A \emph{preference profile} for player $i$, denoted by $\succ_i$, is defined as a complete, reflexive, and transitive binary
relation between the elements of a given set. Within the context of our proposed cache-enabled HO problem, we define the matching problem as follow:\vspace{-.5em}
\begin{definition}\label{def:1}
Given the two disjoint sets of MUEs and BSs, respectively, in  $\mathcal{U}$ and $\mathcal{K}'=\mathcal{K}\cup \{k_0\}$, a single-period \emph{HO matching} is defined as a many-to-one mapping $\mu: \mathcal{U} \cup \mathcal{K}' \rightarrow \mathcal{U} \cup \mathcal{K}'$ that satisfies:
\begin{enumerate}
	\item[1)] $\forall u \in \mathcal{U}$, $\mu(u) \in \mathcal{K}' \cup \{u\}$. In fact, $\mu(u)=k$ means $u$ is assigned to $k$, and $\mu(u)=u$ indicates that the MUE $u$ is not matched to any BS, and thus, will use the cached content.
	\item[2)] $\forall k \in \mathcal{K}'$, $\mu(k) \subseteq \mathcal{U} \cup \{k\}$, and $\forall k \in \mathcal{K}$, $|\mu(k)|\leq U_{k}^{\text{th}}$. In fact, $\mu(k)=k$ implies that no MUE is assigned to the BS $k$.
	\item[3)] $\mu(u)=k$, if and only if $u \in \mu(k)$.\vspace{-1em}
\end{enumerate}
\end{definition}
Note that, by definition, the matching game satisfies constraints \eqref{opt:d}-\eqref{opt:f}. More interestingly, the matching framework allows defining relevant utility functions per MUE and SBSs, which can capture the preferences of MUEs and SBSs. In this regard, the utility that an arbitrary MUE $u \in \mathcal{U}$ assigns to an SBS $k \in \mathcal{K}$ will be:
\begin{align}\label{utility1}
	\Phi(u,k) = 	P_u^{\text{th}}-\mathbbm{P}\left(\sum_{k \in \mathcal{K}}\zeta(u,k)D_{u,k} < v_ut_{\text{MTS}} \right)=P_u^{\text{th}}-\frac{2}{\pi}\arcsin\left(\frac{v_ut_{\text{MTS}}}{2a_k}\right).
\end{align}
Here, we observe that the utility in  \eqref{utility1} is larger for SBSs having a larger cell radius $a_k$. In addition, as the speed of the MUEs increases, the utility generated from those MUEs being assigned to an SBS decreases. 
Meanwhile, the utility that an SBS $k$ assigns to an MUE $u$ is given by\vspace{-1em}
\begin{align}
\Gamma(u,k) = T_s-\frac{\Omega_u}{Q}.
\end{align}
In fact, an SBS assigns higher utility to MUEs that are not capable of using caching for the next time duration $T_s$. Based on the defined utility functions, the preference profile of an arbitrary MUE $u$, $\succ_u$, will be:\vspace{-1em}
\begin{subequations}
\begin{IEEEeqnarray}{rCl}\label{prefer1}
	k \succ_u k' \,\,\,&&\Leftrightarrow \,\,\, \Phi(u,k) > \Phi(u,k'), \label{prefer1-1}\\
	u \succ_u k \,\,\, &&\Leftrightarrow \,\,\, \Phi(u,k)<0, \label{prefer1-2} 
\end{IEEEeqnarray}
\end{subequations}
where $k \succ_u k'$ implies that SBS $k$ is strictly more preferred than SBS $k'$ by MUE $u$. Moreover, $u \succ_u k$ means that an SBS $k$ is not acceptable to an MUE $u$, if and only if the assigned utility is negative. In fact, \eqref{prefer1-2} known as \emph{an individual rationality constraint} and is in line with satisfying the feasibility condition in \eqref{opt:b}. Similarly, we can define the preference profile of an SBS $k$, $\succ_k$, as follows \vspace{-1em}
\begin{subequations}
\begin{IEEEeqnarray}{rCl}\label{prefer2}
	u \succ_k u' \,\,\,&&\Leftrightarrow \,\,\, \Gamma(u,k) > \Gamma(u',k), \label{prefer2-1}\\
	k \succ_k u \,\,\, &&\Leftrightarrow \,\,\, \Gamma(u,k)<0 \label{prefer2-2}, 
\end{IEEEeqnarray}
\end{subequations}
where \eqref{prefer2-2} is the individual rationality requirement for SBSs which is equivalent to satisfying the feasibility constraint in \eqref{opt:c}. With this in mind, the proposed matching game is formally defined as a tuple $\Pi\triangleq (\mathcal{U}\cup \mathcal{K},\boldsymbol{\succ_u},\boldsymbol{\succ_k})$, where $\boldsymbol{\succ_u}=\{\succ_u\}_{u \in \mathcal{U}}$ and $\boldsymbol{\succ_k}=\{\succ_k\}_{k \in \mathcal{K}}$.

To solve this game, one desirable solution concept is to find a \emph{two-sided stable matching} between the MUEs and SBSs, $\mu^*$, which is defined as follow \cite{Roth92}:\vspace{-.5em}
\begin{definition}
	An MUE-SBS pair $(u,k) \notin \mu$ is said to be a \textit{blocking pair} of the matching $\mu$, if and only if $k \succ_{u} \{\mu(u),u\}$ and $u \succ_k \{\mu(k),k\}$.
	Matching $\mu$ is \textit{stable}, $\mu \equiv \mu^*$, if there is no blocking pair.\vspace{-1em}
\end{definition}
A two-sided stable association between MUEs and SBSs ensures fairness for the MUEs. That is, if an MUE $u$ envies the association of another MUE $u'$, then $u'$ must be preferred by the SBS $\mu^*(u')$ to $u$, i.e., the envy of MUE $u$ is not justified.  
\begin{remark}
	For a given single-period HO matching game $\Pi$, the \emph{deferred acceptance (DA)} algorithm \cite{Gale}, presented in Algorithm \ref{algo:1}, is guaranteed to find a two-sided stable association $\mu^*$ between MUEs and SBSs.\vspace{-1em}
\end{remark}
\begin{algorithm}[!t]
	\small
	\caption{DA Algorithm for Single-period Association Between MUEs and SBSs}\label{algo:1}
	\textbf{Inputs:}\,\, $\Pi\triangleq (\mathcal{U}\cup \mathcal{K},\boldsymbol{\succ_u},\boldsymbol{\succ_k})$.\\
	\textbf{Outputs:}\,\, Stable matching $\mu^*$.
	\begin{algorithmic}[1]
		\State  If not already accepted by an SBS, each unmatched MUE $u \in \mathcal{U}$ applies for its most preferred SBS $k \succ_u u$. Remove $k$ from $u$'s preference profile $\succ_u$.
		\State  Each SBS $k \in \mathcal{K}$ receives the proposals from the applicants in Step 1, tentatively accepts  $U_{k}^{\text{th}}$ of most preferred MUEs from new applicants and the MUEs that are so far accepted in $\mu(k)$, and rejects the rest.
		\Repeat \,\,\,Steps $1$ to $2$ \Until{Each MUE $u$ is accepted by an SBS, or $u$ is applied for all SBSs $k \succ_u u$.} 
		\If{$\exists u \in \mathcal{U}, \mu(u) \notin \mathcal{K}$ and ${\Omega_u}/{Q}< T_s$,}
		\State $\mu(u)=u$,
		\Else 
		\State{Assign $u$ to the MBS.}
		\EndIf
	\end{algorithmic}\label{Algorithm1}
\end{algorithm}
Unfortunately, the DA algorithm is not suitable to capture the dynamics of the system which arise from the mobility of the MUEs. In fact, the preference profiles of the MUEs and SBSs only depend on the current state of the system, such as the location of the MUEs, and the cache sizes. In addition, the DA algorithm cannot guarantee stability, if the preference of the MUEs change after HOFs. Thus, to be able to achieve stability for dynamic settings, such as in Scenarios 1 and 2, we need to incorporate the post-HO scenarios into the matching game, such that no MUE can block the stability even after experiencing an HOF. To this end, we extend the notion of one-stage stability in Algorithm \ref{algo:1} into a \emph{dynamic stability} concept that is suitable for the problem at hand. 
\vspace{-1em}
\subsection{Dynamic matching for mobility management in heterogeneous networks}
To account for possible scenarios that may occur after HO, we consider a two-stage dynamic matching game that incorporates within the preference profiles, some of the possible scenarios that may face the MUEs and base stations after handover execution. Such a dynamic matching will allow the MUEs to build preference profiles over different \emph{association plans} rather than SBSs. An association plan is defined as a sequence of two matchings for a given MUE or SBS. For example, $kk'$ is an association plan that indicates an MUE will be assigned to the SBS $k$ followed by another HO to SBS $k'$. In this regard, $k_1k_2 \succ_u k'_1k'_2$ means that MUE $u$ prefers plan $k_1k_2$ to $k'_1k'_2$. With this in mind, we can modify the one-period matching in Definition \ref{def:1} to a relation $\mu^{\dag}: \mathcal{U} \cup \mathcal{K}' \rightarrow (\mathcal{U} \cup \mathcal{K}')^2$, such that $\mu^{\dag}(u)=(\mu_1(u),\mu_2(u))$, where $\mu_1$ and $\mu_2$ are one-period matchings. For example, $\mu^{\dag}(u)=(k,u)$ indicates that MUE $u$ will first perform an HO to SBS $k$, $\mu_1(u)=k$,  followed by using the content of the cache after exiting the coverage of SBS $k$, $\mu_2(u)=u$. Next, we use the following definitions to formally define the stability in dynamic matchings \cite{harvard16}: \vspace{-1em}

\begin{definition}
	An MUE-BS pair $(u,k)$ can \emph{period-1 block} the matching, if any of the following conditions is satisfied: 1) $kk \succ_u \mu^{\dag}(u)$ and $uu \succ_k \mu^{\dag}(k)$; 2) $ku \succ_u \mu^{\dag}(u)$ and $uk \succ_k \mu^{\dag}(k)$; 3) $uk \succ_u \mu^{\dag}(u)$ and $ku \succ_k \mu^{\dag}(k)$; or 4) $uu \succ_u \mu^{\dag}(u)$ and $kk \succ_k \mu^{\dag}(k)$.	
	A matching is \emph{ex ante stable}, if it cannot be period-1 blocked by any MUE/BS or MUE-BS pair. \vspace{-1em}
\end{definition}
 In a dynamic matching problem, either the  MUEs or the BSs may block the matching, after knowing the outcome of the first matching $\mu_1$. In this regard, we define the notion of period-2 blocking and dynamic stability as follows:\vspace{-1em}
\begin{definition}
An MUE $u$ can \emph{period-2 block} a matching $\mu^{\dag}$ if $(\mu_1(u),u)\succ_u \mu^{\dag}(u)$. Similarly, an MUE-BS pair $(u,k)$ can \emph{period-2 block} if any of the following conditions is satisfied: 1) $(\mu_1(u),k)\succ_u \mu^{\dag}(u)$ and $(\mu_1(k),u)\succ_k \mu^{\dag}(k)$, or 2) $(\mu_1(u),u)\succ_u \mu^{\dag}(u)$ and $(\mu_1(k),k)\succ_k \mu^{\dag}(k)$. A matching is said to be \emph{dynamically stable}, if it cannot be period-1 or period-2 blocked by any MUE or MUE-BS pair\footnote{In general, a matching is dynamically stable for any time $t$, if it cannot be period-$t$ blocked by any MUE or MUE-BS pair. Extending the dynamic matching to more than two periods depends on how much information is available for MUEs about the network. In this work, we focus on a two-period matching problem, since it is more tractable and practical.}.\vspace{-1em}
\end{definition}
From Definitions 3 and 4, we can see that, any dynamically stable matching is also an ex ante stable matching. However, ex ante stability does not guarantee dynamic stability. For example, if $\mu^{\dag}(u)=(k,u)$ for an MUE $u$, ex ante stability does not guarantee that the MUE commits to use the cache, if the first handover to SBS $k$ fails. In other words, the MUE may block an ex ante stable matching after the actual outcome of the first matching is known. To help better understand the stability for dynamic matchings, we consider the following simple example.\vspace{-1em}
\begin{example}\label{example1}
	Consider a dynamic matching game $\Pi^{\dag}$, composed of MUEs $\mathcal{U}=\{u_1,u_2\}$, MBS $k_0$, and SBSs $\mathcal{K}=\{k_1,k_2\}$, with $U_k^{\text{th}}=1$ for $k=k_1,k_2$, as shown in scenario 2 of Fig. \ref{scenarios}. The preference plans of MUEs, MBS $k_0$, and SBSs are as follows:\vspace{-1em}
	\begin{align*}
		&\succ_{u_1}: k_1k_0, \underline{k_1u_1},u_1k_0,u_1u_1; &&\succ_{u_2}:k_1u_2, \underline{u_2k_2}, u_2u_2;\\
		&\succ_{k_1}: \underline{u_1k_1},u_2k_1,k_1k_1;\hspace{0.2cm} && \succ_{k_2}: \underline{k_2u_2}, k_2k_2; &&\succ_{k_0}: k_0u_1, \underline{k_0k_0};
			\end{align*}
	where the preference profiles are sorted in descending order and association plans that are not included do not meet the individual rationality constraint. Here,  the underlined matching is one of the possible ex ante stable matchings. However, this matching is not dynamically stable. That is because conditioned to $\mu_1(u)=k_1$, the MUE-MBS pair $(u_1,k_0)$ will period-2 block the matching, since $k_1k_0 \succ_{u_1} k_1u_1$ and $k_0u_1 \succ_{k_0} k_0k_0$. In practice, such a blocking occurs if the MUE experiences an HOF with its first matching to $k_1$.
\end{example} 
\vspace{-1em}
Next, we propose an algorithm that finds a dynamically stable solution for the proposed mobility management problem.
\vspace{-1em}
\subsection{Dynamically stable matching algorithm for mobility management}
To find the dynamically stable solution, we note that the solution must first admit the ex ante stability. Therefore, we propose an algorithm, inspired from \cite{harvard16} that yields an ex ante stable association in the first stage, followed by a simple modification to resolve any possible period-2 blocking cases. For each MUE $u$, let $\mathcal{P}_u=\cup_{k \in \mathcal{K}}\{kk,uk,ku\}$ be the set of all plans considered by $u$. The algorithm proceeds as follows:\\
\textbf{Stage-1 (Finding an ex ante stable matching):}
\begin{enumerate}
	\item For each MUE $u\in\mathcal{U}$, if $uu\succ_u \kappa$, for all $\kappa \in \mathcal{P}_u$, then $u$ does not send any plan proposal to the BSs. Otherwise, MUE $u$ sends a plan proposal to a BS, according to the most preferred plan $\kappa_u^*$ as follows. If $\kappa_u^*=kk$, MUE $u$ sends a request for a two-period association to the BS $k$. If $\kappa_u^*=ku$, the MUE sends an association request to BS $k$, only for period-1. Similarly, if $\kappa_u^*=uk$, the MUE sends an association request to $k$ only for period-2. The MUE removes $\kappa_u^*$ from its preference profile for the rest of the procedure.
	\item Each SBS $k \in \mathcal{K}$ receives the plan proposals and tentatively accepts the most preferred plans, such that the quota $U_k^{\text{th}}$ is not violated at each period. Clearly, any accepted plan $\kappa$ by SBS $k$ satisfies $\kappa \succ_k kk$.
	\item MUEs with rejected plans apply in the next round, based on their next most preferred plan. The first stage of the algorithm converges, once no plan is rejected.
\end{enumerate}\vspace{-1em}
\begin{proposition}\label{prop:1}
	Stage-1 of the proposed algorithm in Algorithm \ref{algo:2} converges to an ex ante stable association between MUEs and BSs.
\end{proposition}\vspace{-1.5em}
\begin{proof}
	Assume an MUE-BS pair $(u,k)$ period-1 blocks the matching $\mu^{\dag}$. In consequence, there is a plan $\kappa \in \{ku,uk,kk\}$ for $u$ and a corresponding plan for $k$ that both prefer to their current matching in $\mu^{\dag}$. If $\kappa \succ_{u} \mu^{\dag}(u)$, then the MUE $u$ must have sent a proposal for $\kappa$ to $k$ prior to its associated plan in $\mu^{\dag}$. Since $\kappa$ is not eventually accepted, that means at some point, the SBS $k$ has rejected $\kappa$ in favor of another plan. Since the matching for SBSs improves at each round, we conclude that $\kappa$ is less preferred by $k$ compared to $\mu^{\dag}(k)$. This contradicts the first assumption, thus, such a period-1 blocking pair does not exist and $\mu^{\dag}$ is ex ante stable.
\end{proof}

To avoid period-2 blockage, we introduce a certain structure to the preference profile of the SBSs as follows. For any SBS for whom the maximum quota of $U_k^{\text{th}}$ MUEs are assigned, i.e. $|\mu_2^{\dag}(k)|=U_{k}^{\text{th}}$,\vspace{-1em}
\begin{align}\label{cond1}
	\mu^{\dag} \succ_k \left(\mu_1^{\dag}(k),\tilde{\mu_2}^{\dag}(k)\cup \{u\}\right),
\end{align}
where $\tilde{\mu_2}^{\dag}(k)$ is ${\mu_2}^{\dag}(k)$ with one associated MUE removed to accommodate a new matching with MUE $u$. In fact, \eqref{cond1} implies that an MUE cannot period-2 block the matching with any SBS $k$ that is associated to $U_k^{\text{th}}$ MUEs. In addition, \vspace{-.5em}
\begin{align}\label{cond2}
	  (\mu_1(k_0),u)\succ_{k_0} \mu^{\dag}\iff P_{\mu_1^{\dag}(u)}^{\text{th}}-\frac{2}{\pi}\arcsin\left(\frac{v_ut_{\text{MTS}}}{2a_{\mu_1^{\dag}(u)}}\right)<\epsilon,
\end{align}
where $ \epsilon$ is a non-negative scalar. In fact, \eqref{cond2} allow MUEs that are assigned to SBSs in period 1, with not small enough HOF probability, to be assigned to the MBS in period 2. Another alternative was to set $P^{\text{th}}$ a small value from the start. However, this policy will discourage MUEs to be assigned to SBSs and could increase the load on the MBS. With this in mind, we construct the second stage of the algorithm as follows:\\
\textbf{Stage-2 (Remove period-2 blocking pairs):} Apply the deferred acceptance algorithm shown in Algorithm \ref{algo:1} to a subset of MUEs with $\mu_2^{\dag}(u)=u$, and subset of BSs with $|\mu_2^{\dag}(k)|<U_k^{\text{th}}$, while considering the constraints in \eqref{cond1} and \eqref{cond2}.
\begin{algorithm}[!t]
	\small
	\caption{Proposed Algorithm for Dynamic Matching Between MUEs and BSs}
	\textbf{Inputs:}\,\, Preference plans $\kappa$ for all MUEs, MBS, and SBSs.\\
	\textbf{Outputs:}\,\, Dynamically stable matching $\mu^*$.
	\begin{algorithmic}[1]
		\item[]\textit{Phase 1:}
		\State  For each MUE $u\in\mathcal{U}$, if $uu\succ_u \kappa$, for all $\kappa \in \mathcal{P}_u$, then $u$ does not send any plan proposal to the BSs. Otherwise, MUE $u$ sends a plan proposal to a BS, according to the most preferred plan $\kappa_u^*$.
		\State  Each SBS $k \in \mathcal{K}$ receives the plan proposals and tentatively accepts most preferred plans (also compared to plans that are previously accepted), such that the quota $U_k^{\text{th}}$ is not violated at each period. Clearly, any accepted plan $\kappa$ by SBS $k$ satisfies $\kappa \succ_k kk$.
		\Repeat \,\,\,Steps $1$ to $2$ \Until{No plan is rejected. The yielded ex ante stable matching is denoted by $\mu^{\dag}=(\mu_1^{\dag},\mu_2^{\dag})$.} \newline
		\noindent\textit{Phase 2:}
		\If{$\exists u \in \mathcal{U}, \mu_2^{\dag}(u) = u$,}
		apply DA algorithm in Algorithm \ref{algo:1} to the subset of MUEs with $\mu_2^{\dag}(u) = u$ and the subset of BSs with $|\mu_2^{\dag}(k)|<U_k^{\text{th}}$, considering the constraints in \eqref{cond1} and \eqref{cond2}. Return yielded matching.
		\Else 
		\State return $\mu^{\dag}$.
		\EndIf
	\end{algorithmic}\label{algo:2}
\end{algorithm}
The proposed two-stage algorithm is summarized in Algorithm  \ref{algo:2}. Reconsidering Example \ref{example1}, it is easy to follow that Algorithm \ref{algo:2} yields the following solution which is dynamically stable\footnote{Here, we assume that \eqref{cond2} holds for $u_1$. Otherwise, the ex ante stable solution in Example \ref{example1} is also dynamically stable, since $k_0$ will not make a period-2 block pair with $u_1$.}:\vspace{-1em}
	\begin{align*}
&\succ_{u_1}: \underline{k_1k_0}, k_1u_1,u_1k_0,u_1u_1; &&\succ_{u_2}:k_1u_2, \underline{u_2k_2}, u_2u_2;\\
&\succ_{k_1}: \underline{u_1k_1},u_2k_1,k_1k_1;\hspace{0.2cm} && \succ_{k_2}: \underline{k_2u_2}, k_2k_2; &&\succ_{k_0}:  \underline{k_0u_1},k_0k_0.
\end{align*}
For the proposed algorithm, we can state the following results:\vspace{-1em}
\begin{theorem}
	The proposed two-stage algorithm in Algorithm \ref{algo:2} is guaranteed to converge to a dynamically stable association between MUEs and BSs.
\end{theorem}\vspace{-1.5em}
\begin{proof}
	From Proposition \ref{prop:1}, the solution is guaranteed to be ex ante stable. Therefore, MUEs and BSs will not period-1 block the matching. The rest of the proof easily follows the fact that the BSs will not make a period-2 blocking pair with any MUE, due to the constraints in \eqref{cond1} and \eqref{cond2}. In fact, if there is any period-2 blocking pair $(u,k)$, there are four possible cases to consider: 1) $kk \succ_u \mu^*(u)$ and $uu \succ_k \mu^*(k)$,
		2) $uk \succ_u \mu^*(u)$ and $ku \succ_k \mu^*(k)$,
		3) $uk \succ_u \mu^*(u)$ and $u'u \succ_k \mu^*(k)$, where $u' \neq u$, or
		4) $k'k \succ_u \mu^*(u)$ and $u'u \succ_k \mu^*(k)$, where $k' \neq k$ and $u' \neq u$.
	The first two cases are not possible, since they indicate that $(u,k)$ can period-1 block $\mu^*$ which contradicts ex ante stability. Considering the last two cases, since MUE $u$ is not associated with SBS $k$ in period-2, that implies that $k$ has already been assigned to $U_{k}^{\text{th}}$ MUEs. Otherwise, $u$ would be assigned to $k$ during the second stage of Algorithm \ref{algo:2}. Hence, due to the constraint \eqref{cond1}, $k$ will not make a period-2 blocking pair with $u$. Similarly, the MBS will not make a period-2 blocking pair with any MUE $u$ that is not assigned to the MBS during the second stage. Thus, no period-2 blocking pair exists and the solution $\mu^*$ satisfies dynamic stability. 
\end{proof}\vspace{-0em}

To analyze the signaling overhead of the proposed algorithm, we consider the total number of HO requests sent to a target SBS by the MUEs. Additional control signals from the SBSs to MUEs can be managed by using a broadcast channel and do not significantly contribute to the overhead of the proposed scheme. In this regard, consider the worst-case scenario in which the initial cache size is $\Omega_u=0$ for all $u \in \mathcal{U}$. Therefore, all MUEs seek to perform an HO to the target SBS $k$ by sending a request for plan $\kappa = ku$ during Stage-1 of the proposed algorithm. The SBS $k$ accepts up to $U_k^{\text{th}}$  association plans and rejects the rest. Clearly, if there is one target SBS for the MUEs, the signaling overhead will be $\mathcal{O}(U)$. Otherwise, rejected MUEs will send an HO request to the next target SBS, based on their preference profiles. The maximum signaling overhead occurs for a case when all MUEs have the same preference profile as it introduces the highest competition among MUEs. In this case, the signaling overhead of the proposed algorithm will be $\mathcal{O}(UK)$. In addition, in Section \ref{sec:sim}, we will discuss how caching capabilities will reduce the overhead of the proposed algorithm. 

\begin{table}[t!]
	\scriptsize
	\centering
	\caption{
		\vspace*{-0em}Simulation parameters}\vspace*{-1em}
	\begin{tabular}{|c|c|c|}
		\hline
		\bf{Notation} & \bf{Parameter} & \bf{Value} \\
		\hline
		$f_c$ & Carrier frequency & $73$ GHz\\
		\hline
		$P_{t,k}$ & Total transmit power of SBSs & $\left[20, 27, 30\right]$ dBm\\
		\hline
		$K$ & Total number of SBSs & $50$\\
		\hline
		$w$ & Available Bandwidth & $5$ GHz\\
		\hline
		($\alpha_{\text{LoS}}$,$\alpha_{\text{NLoS}}$) & Path loss exponent& ($2,3.5$) \cite{Ghosh14}\\
		\hline
		$d_0$ & Path loss reference distance& $1$ m \cite{Ghosh14}\\
		\hline
		$G_{\text{max}}$ & Antenna main lobe gain& $18$ dB\cite{7110547} \\
		\hline
		$G_{\text{min}}$ & Antenna side lobe gain& $-2$ dB\cite{7110547} \\
		\hline
		$N_k$ & Number of mmW beams& $3$ \\
		\hline
		$\theta_{m}, \theta_k$ & beam width& $10^{\circ}$\cite{7110547} \\
		\hline
		$N_0$ & Noise power spectral density& $-174$ dBm/Hz\\
		\hline
		$t_{\text{MTS}}$ & Minimum time-of-stay& $1$s \cite{3gpp} \\
		\hline
		$Q$ & Play rate& $1$k segments per second \\
		\hline
		$B$ & Size of video segments& $1$ Mbits  \\
		\hline
		$(v_{\text{min}},v_{\text{max}})$ & Minimum and maximum MUE speeds& $(1,16)$ m/s  \\
		\hline
		$E^s$ & Energy per inter-frequency scan& $3$ mJ \cite{power}  \\
		\hline
	\end{tabular}\label{tab1}\vspace{-.5em}
\end{table}
\vspace{-.6cm}
\section{Simulation Results}\label{sec:sim}
For simulations, we consider a HetNet composed of $K=50$ SBSs distributed uniformly across a circular area with radius $500$ meters with a single MBS located at the center and a minimum inter-cell distance of $30$ meters. Moreover, the transmit power of SBSs are chosen randomly from the set of powers in $\left[20,27,30\right]$ dBm. The main parameters are summarized in Table \ref{tab1}. In our simulations, we consider the overall transmit-receive antenna gain from an interference link to be random. All statistical results are averaged over a large number of independent runs. Next, we first investigate the gains achievable by the proposed cache-enabled scheme for a single user scenario. Then, we will evaluate the performance of the proposed dynamic matching approach by extending the results for scenarios with multiple MUEs in which SBSs can only serve a limited number of MUEs. 
\begin{figure}[!t]
	\centering
	\centerline{\includegraphics[width=10cm]{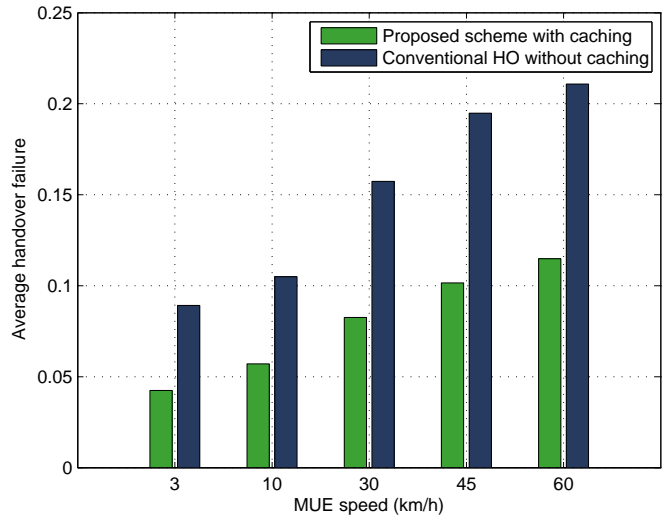}}\vspace{-.4cm}
	\caption{\small HOF vs different MUE speeds.}\vspace{-0.3cm}
	\label{main1sim}
\end{figure}\vspace{0em}
\vspace{-.7cm}
\subsection{Analysis of the proposed cache-enabled mobility management for single user scenarios}
Fig. \ref{main1sim} compares the average HOF of the proposed scheme with a conventional HO mechanism without caching. The results clearly demonstrate that caching capabilities, as proposed here, will significantly improve the HO process for dense HetNets. In fact, the results in Fig. \ref{main1sim} show that caching over mmW frequencies will reduce HOF for all speeds, reaching up to $45 \%$ for MUEs with $v_u=60$ km/h.

\begin{figure}[!t]
	\centering
	\centerline{\includegraphics[width=10cm]{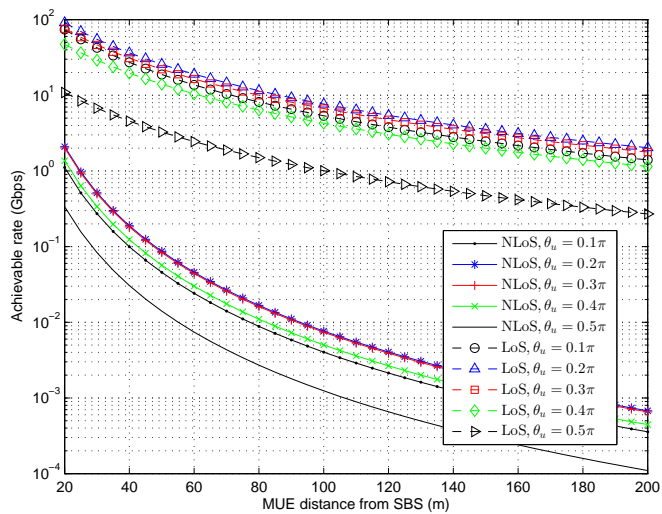}}\vspace{-.4cm}
	\caption{\small Achievable rate of caching vs $r_{u,k}(\boldsymbol{x})$ for different $\theta_u$.}\vspace{-1em}
	\label{rate}
\end{figure}\vspace{0em}
 
Fig. \ref{rate} shows the achievable rate of caching for an MUE with $v_u=60$ km/h, as a function of different initial distances $r_{u,k}(\boldsymbol{x})$ for various $\theta_u$. The results in Fig. \ref{rate} show that even for MUEs with high speeds, the achievable rate of caching is significant, exceeding $10$ Gbps, for all $\theta_u$ values and inital distance of $20$ meters from the SBS. However, we can observe that the blockage can noticeably degrade the performance. In fact, for NLoS scenarios,  the maximum achievable rate at a distance of $20$ meters decreases to $2$ Gbps.
\vspace{-1em}
\subsection{Performance of the proposed dynamically stable mobility management algorithm}
Here, we consider the set of MUEs entering a target cell coverage region with random directions and speeds. Moreover, the cache sizes of the MUEs are initially $\Omega_u=10^4$ segments for all MUEs. In addition, each SBS can serve up to $U_{k}^{\text{th}}=10$ MUEs. Depending on the speed of the MUE, its direction, and the location of the next target SBS, MUEs form their preferences over different plans as elaborated in Section \ref{sec:V}.  
\begin{figure}[!t]
	\centering
	\centerline{\includegraphics[width=10cm]{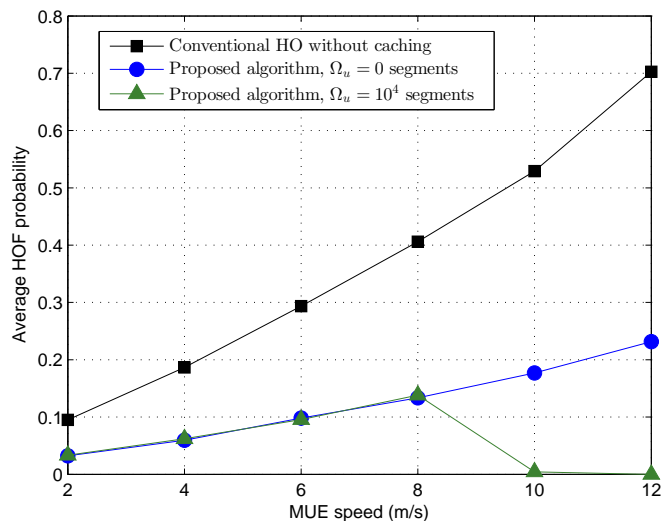}}\vspace{-.5cm}
	\caption{\small Average HOF probability versus MUEs' speeds.}\vspace{-.8cm}
	\label{sim3}
\end{figure}\vspace{0em}

In Fig. \ref{sim3}, the average HOF probability of the proposed algorithm is compared with a conventional scheme that does not incorporate caching, versus the speed of the MUEs. The HOF probability is defined as the ratio of the MUEs with HOF to the total number of MUEs, for $U=20$ and $U_{k}^{\text{th}}=10$. The results in Fig. \ref{sim3} show that the HOF probability increases with the speed of the MUEs, since the ToS will decrease for higher MUE speeds. In addition, we observe that the proposed algorithm can significantly reduce the HOF probability by leveraging the information on the MUE's trajectory and the network's topology. Fig. \ref{sim3} also shows that for a non-zero initial cache sizes of $\Omega_u=10^4$ segments, the algorithm is considerably robust against HOF. In fact, the HOF probability declines for speeds beyond $v_u=8$ m/s, since higher speed allows the MUE to traverse larger distance before the cached video segments run out. Therefore, more MUEs will be able to skip an HO to the target cell and use the cached content to move to the next available SBS. 

\begin{figure}[!t]
	\centering
	\centerline{\includegraphics[width=10cm]{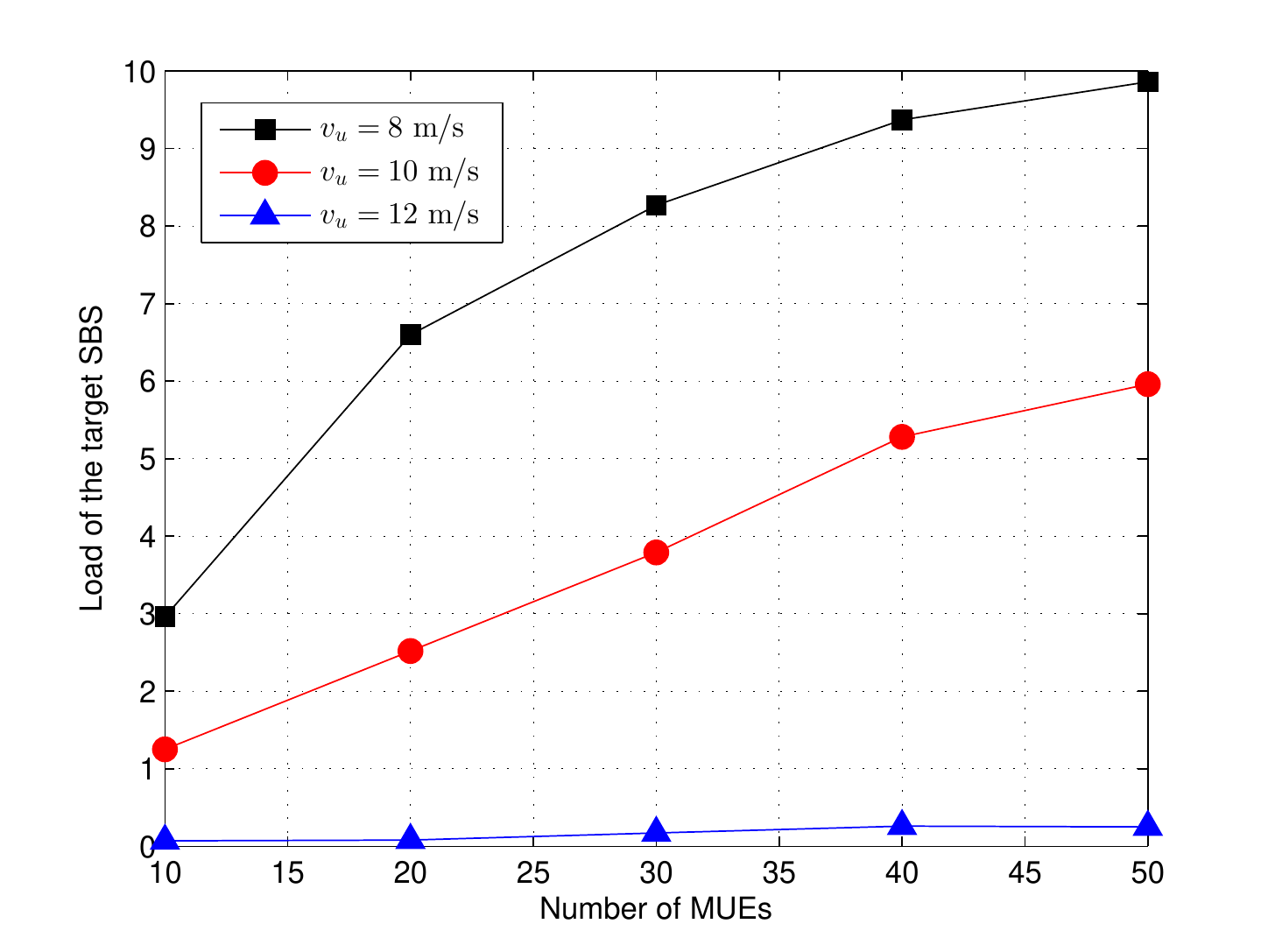}}\vspace{-.4cm}
	\caption{\small Load of the target SBS vs the number of MUEs.}\vspace{-0.2cm}
	\label{load1}
\end{figure}\vspace{0em}
Fig. \ref{load1} shows the load of the target cell versus the number of MUEs for different MUE speeds $v_u=8, 10,$ and $12$ m/s, SBS quota $U_{k}^{\text{th}}=10$, and initial cache size $\Omega_u=10^4$ segments. Here, we observe that the proposed algorithm associates less MUEs to the target cell as the speed increases. That is due to two reasons: 1) higher speeds decrease the ToS and increase the chances of HOFs, and 2) with higher speeds, MUEs can traverse longer distances by using $\Omega_u$ cached segments and it is more likely that they can reach to the next target SBS. Fig. \ref{load1} shows that the load of the target cell reduces up to $45 \%$ when $v_u$ increases from $8$ to $10$ m/s for $U=40$~MUEs.

\begin{figure}[!t]
	\centering
	\centerline{\includegraphics[width=10cm]{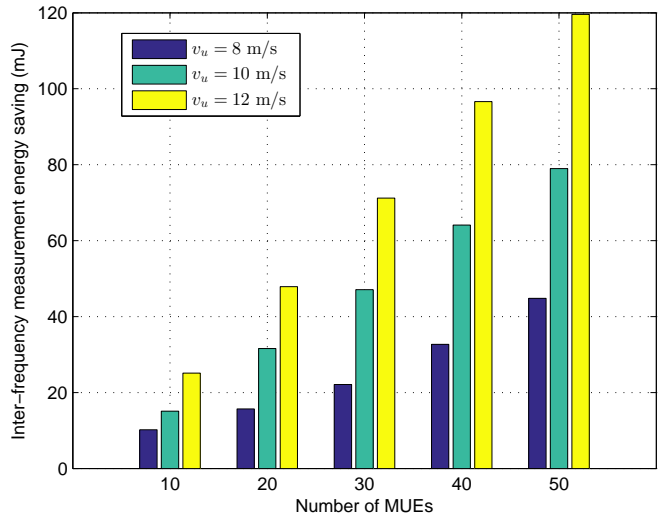}}\vspace{-.4cm}
	\caption{\small Energy savings for inter-frequency measurements vs number of MUEs.}\vspace{-0.2cm}
	\label{load2}
\end{figure}\vspace{0em}

In Fig. \ref{load2}, the inter-frequency measurement energy savings yielded by the proposed algorithm are shown as a function of the number of MUEs. Fig. \ref{load2} shows the total saved energy for MUEs that will use the cached content and do not perform any inter-frequency measurements for handover to an SBS for an initial cache size of $\Omega_u= 10^4$ segments and different MUE speeds. For $U=50$,  MUEs that perform conventional handover without caching will require $UE^s=150$ mJ total energy for performing inter-frequency measurements. However, the results in Fig. \ref{load2} show that the proposed scheme achieves up to $80 \%$, $52 \%$, and $29 \%$ gains in saving energy, respectively, for MUE speeds $v_u=8, 10,$ and $12$ m/s by leveraging cached segments and muting unnecessary cell search. Given that the required energy for measurements linearly scales with the number of MUEs, the results in Fig. \ref{load2} can also be interpreted as the offloading gains of the proposed approach, compared with conventional HO with no caching. Moreover, these results are consistent with those shown in Fig. \ref{load1}. In fact, as the speed of MUEs increases, the HOF probability increases, and thus, MUEs tend to be assigned to the MBS or use their cached content. In addition, fast moving MUEs are more likely to reach the next target cell before the cached content runs out.


\begin{figure}[!t]
	\centering
	\centerline{\includegraphics[width=10cm]{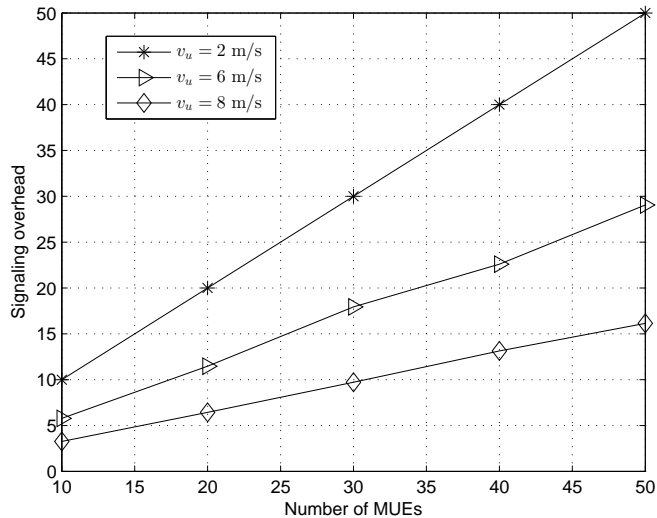}}\vspace{-.5cm}
	\caption{\small Signaling overhead vs number of MUEs.}\vspace{-0.2cm}
	\label{sim4}
\end{figure}\vspace{0em}

In Fig. \ref{sim4}, we show the signaling overhead resulting from the proposed algorithm versus the number of MUEs, for $\Omega_u=10^4$ initial cache size and different MUE speeds. Here, we refer to the signaling overhead as the number of HO requests sent to the target SBS by the MUEs. Fig. \ref{sim4} shows that for low speeds $v_u=2$ m/s, almost all MUEs will attempt to hand over to the target SBS, since the time needed for traversing the SBS coverage is longer than the time available by using the cached content. Nonetheless, the results in Fig. \ref{sim4} clearly demonstrate that the proposed algorithm has a manageable overhead, not exceeding $17$ requesting signals for a network size of $U=50$ with $v_u=8$ m/s. In fact, it is interesting to note that although mobility management is, in general, more challenging for high speed MUEs, the overhead of the proposed algorithm decreases for high speed scenarios. This is due to the fact that high speed MUEs use the cached content more effectively than slow-moving MUEs. \vspace{-1em}
\section{Conclusions}\label{conclude}\vspace{-.2cm}
In this paper, we have proposed a comprehensive framework for mobility management in integrated microwave-millimeter wave cellular networks. In particular, we have shown that by smartly caching video contents while exploiting the dual-mode nature of the network's base stations, one can provide seamless mobility to the users. We have derived various fundamental results on the probability and the achievable rate for caching video contents by leveraging millimeter wave high capacity transmissions. In addition, to capture the dynamics of the mobility management, we have formulated the multi-user handover problem as a dynamic matching game between the mobile users and small base stations. To solve this game, we have proposed a novel algorithm that is guaranteed to converge to a dynamically stable handover mechanism. Moreover, we have shown that the proposed cache-enabled mobility management framework provides significant gains in reducing the number of handovers, energy consumption for inter-frequency scanning, as well as mitigating the handover failure. Numerical results have corroborated our analytical results and showed that the significant rates for caching can be achieved over the mmW frequencies, even for fast mobile users. In addition, the results have shown that the proposed approach substantially decreases the handover failures and provides significant energy savings in heterogeneous networks.

\vspace{-.5cm}
\def\baselinestretch{0.88}
\bibliographystyle{ieeetr}
\bibliography{references}
\end{document}